\newif\ifabstract
\abstracttrue
 \abstractfalse 
\newif\iffull
\ifabstract \fullfalse \else \fulltrue \fi

\documentclass[11pt]{article}
\usepackage{amsfonts}
\usepackage{amssymb}
\usepackage{amstext}
\usepackage{amsmath}
\usepackage{xspace}
\usepackage{theorem}
\usepackage{graphicx}
\usepackage{graphics}
\usepackage{colordvi}
\usepackage{colordvi}
\usepackage{enumitem}
\ifabstract
\usepackage[small,compact]{titlesec}
\fi

\usepackage{wrapfig}
\usepackage{hyperref}

\advance \oddsidemargin by -0.8in

\newcommand{\tw}{\mathrm{tw}}

\newenvironment{proofof}[1]{\noindent{\bf Proof of #1.}}%
        {\hspace*{\fill}$\Box$\par\vspace{4mm}}

\newcounter{definition}
\newenvironment{definition}{\refstepcounter{definition}\par\medskip\noindent%
   \textbf{Definition~\thedefinition.} \rmfamily}{\medskip}

\newcommand{\qed}{\hfill\vbox{\hrule height.2pt\hbox{\vrule width.2pt height5pt \kern5pt
\vrule width.2pt} \hrule height.2pt}}

\newcommand{\tH}{\tilde{H}}

\newcommand{\alphawl}{\ensuremath{\alpha_{\mbox{\tiny{\sc WL}}}}}
\newcommand{\alphaWL}{\alphawl}

\newcommand{\gkrv}{\ensuremath{\gamma_{\mbox{\tiny{\sc CMG}}}}}
\newcommand{\gKRV}{\gkrv}

\newcommand{\gammaKRV}{\gkrv}

\newcommand{\alphaCMG}{\ensuremath{\alpha_{\mbox{\tiny{\sc CMG}}}}}
\newcommand{\alphaKRV}{\alphaCMG}

\newcommand{\ceil}[1]{\ensuremath{\left\lceil#1\right\rceil}}

\newcommand{\NP}{\mbox{\sf NP}\xspace}
\newcommand{\coNP}{\mbox{\sf coNP}\xspace}

\newcommand{\polylog}[1]{\mathrm{polylog(#1)}}

\newcommand{\set}[1]{\left\{ #1 \right\}}
\newcommand{\sse}{\subseteq}

\newcommand{\tset}{{\mathcal T}}

\newcommand{\pset}{{\mathcal{P}}}

\newcommand{\qset}{{\mathcal{Q}}}
\newcommand{\lset}{{\mathcal{L}}}
\newcommand{\bset}{{\mathcal{B}}}

\newcommand{\cset}{{\mathcal{C}}}

\newcommand{\yset}{{\mathcal{Y}}}
\newcommand{\rset}{{\mathcal{R}}}

\newcommand{\hset}{{\mathcal{H}}}

\newcommand{\sset}{{\mathcal{{S}}}}

\newcommand{\zset}{{\mathcal{Z}}}

\newcommand{\nots}{\overline S}

\newcommand{\be}{\begin{enumerate}}
\newcommand{\ee}{\end{enumerate}}
\newcommand{\bd}{\begin{description}}
\newcommand{\ed}{\end{description}}
\newcommand{\bi}{\begin{itemize}}
\newcommand{\ei}{\end{itemize}}

\newtheorem{theorem}{Theorem}[section]
\newtheorem{lemma}[theorem]{Lemma}

\newtheorem{claim}[theorem]{Claim}

\newtheorem{conjecture}[theorem]{Conjecture}
\newtheorem{question}{Question}

\newenvironment{proof}{\smallskip\noindent{\bf Proof:}}{\hfill\stopproof}
\def\stopproof{\square}
\def\square{\vbox{\hrule height.2pt\hbox{\vrule width.2pt height5pt \kern5pt
\vrule width.2pt} \hrule height.2pt}}

\renewcommand{\phi}{\varphi}
\newcommand{\eps}{\epsilon}

\newcommand{\half}{\ensuremath{\frac{1}{2}}}

\newcommand{\poly}{\operatorname{poly}}

\newcommand{\expect}[2][]{\text{\bf E}_{#1}\left [#2\right]}
\newcommand{\prob}[2][]{\text{\bf Pr}_{#1}\left [#2\right]}

\newcommand{\mincut}{\operatorname{MinCut}}
\newcommand{\out}{\operatorname{out}}

\begin{document}

\title{Degree-3 Treewidth Sparsifiers\footnote{An extended abstract 
of this paper is to appear in {\em Proc.\ of ACM-SIAM Symposium on 
Discrete Algorithms (SODA), January 2015.}}}
\author{
Chandra Chekuri\thanks{Dept.\ of Computer Science, University
of Illinois, Urbana, IL 61801. {\tt chekuri@illinois.edu}.
Supported in part by NSF grant CCF-1319376.}
\and 
Julia Chuzhoy\thanks{Toyota Technological Institute, Chicago, IL
60637. Email: {\tt cjulia@ttic.edu}.
Supported in part by NSF grant CCF-1318242.}
}

\date{\today}
\ifabstract
\begin{titlepage}
\maketitle
\thispagestyle{empty}
\fi
\iffull
\maketitle
\fi

\begin{abstract}
  We study treewidth sparsifiers. Informally, given a graph $G$ of
  treewidth $k$, a treewidth sparsifier $H$ is a minor of $G$, whose
  treewidth is close to $k$, $|V(H)|$ is small, and the maximum vertex
  degree in $H$ is bounded. Treewidth sparsifiers of degree $3$ are of
  particular interest, as routing on node-disjoint paths, and
  computing minors seems easier in sub-cubic graphs than in
  general graphs.

  In this paper we describe an algorithm that, given a graph $G$ of
  treewidth $k$, computes a topological minor $H$ of $G$ such that (i)
  the treewidth of $H$ is $\Omega(k/\polylog k)$; (ii) $|V(H)| =
  O(k^4)$; and (iii) the maximum vertex degree in $H$ is $3$.  The
  running time of the algorithm is polynomial in $|V(G)|$ and $k$.
  Our result is in contrast to the known fact that unless $\NP
  \subseteq \coNP/{\sf poly}$, treewidth does not admit
  polynomial-size kernels.
  One of our key technical tools, which is of independent interest, is
  a construction of a small minor that preserves
  node-disjoint routability between two pairs of vertex subsets. This
  is closely related to the open question of computing small
  good-quality vertex-cut sparsifiers that are also minors of
  the original graph.
\end{abstract}

\ifabstract
\end{titlepage}
\fi

\section{Introduction}
Given a large graph $G$, the goal in graph sparsification is to
compute a ``small'' graph $H$ that retains, exactly or approximately,
some key properties of $G$. Two such standard regimes are when
$V(H)=V(G)$ but $H$ is a sparse graph, or when $|V(H)| \ll |V(G)|$.
Sparsifiers for basic properties such as connectivity, distances, cuts
and flows have been extensively studied. For instance, cut sparsifiers
were introduced by Benczur and Karger \cite{BenczurK96}, and were more
recently generalized to spectral sparsifiers \cite{BatsonSST13}, and
to cut and flow sparsifiers for vertex subsets \cite{Moitra,LM}.
Graph sparsifiers are closely related to the
notion of kernelization used in fixed-parameter tractable algorithms,
where an input instance is first reduced to a much smaller instance
(called a kernel), whose size is ideally polynomial in the parameter
$k$, and then the problem is solved on the smaller instance.
Sparsification and sparse representations are also of great importance
for other objects such as signals, matrices, and geometric objects to
name just a few.

We say that a graph $H$ is a \emph{strong} sparsifier for the given
graph $G$, if additionally $H$ is a minor of $G$. Strong sparsifiers
are of particular interest, since they retain some of the structure of
$G$. For example, if $H$ contains some graph $H'$ as a minor, then so
does $G$; a collection $\pset$ of disjoint paths (or cycles) in $H$
immediately translates to a collection of disjoint paths (or cycles)
in $G$, and so on.

In this paper we study sparsifiers for {\em treewidth}, a fundamental
graph parameter with a wide variety of applications in graph theory
and algorithms. The treewidth of a graph $G=(V,E)$ is typically
defined via tree decompositions.  A tree-decomposition of $G$ consists
of a tree $T=(V(T),E(T))$ and a collection of vertex subsets $\{X_v
\subseteq V\}_{v \in V(T)}$ called bags, such that: (i) for each edge
$(a,b) \in E$, there is some node $v \in V(T)$ with both $a,b \in X_v$
and (ii) for each vertex $a \in V$, the set of all nodes of $T$ whose
bags contain $a$ form a non-empty connected subtree of $T$. The {\em
  width} of a given tree decomposition is $\max_{v \in V(T)} |X_v| -
1$, and the treewidth of a graph $G$, denoted by $\tw(G)$, is the
width of a minimum-width tree decomposition for $G$.  Treewidth is
known to be \NP-hard to compute \cite{treewidth-np-hard}. The best
known polynomial-time approximation algorithm, given a graph $G$ of
treewidth $k$, computes a tree decomposition of width $O(k \sqrt{\log
  k})$ \cite{FeigeHL05}. It is also known that treewidth is
fixed-parameter-tractable~\cite{Bodlaender-tw-fpt}: for every fixed
$k$, there is a linear-time algorithm, that, given $G$, decides
whether $\tw(G) \le k$; the dependence of the running time on $k$ is exponential in $\poly(k)$.  There are many important results on the structure
of large-treewidth graphs. Perhaps the most well-known of these is the
Grid-Minor Theorem of Robertson and Seymour that we discuss in more
detail later.

Informally, graph $H$ is a treewidth sparsifier for a given graph $G$,
if $H$ is sparse, $|V(H)|$ is small, and $\tw(H)$ is (approximately)
the same as $\tw(G)$.  For $H$ to be useful as a replacement for $G$,
it needs to be a strong sparsifier --- that is, $H$ should be a minor
of $G$\footnote{Note that if all we wanted is a graph $H$ that has
  similar treewidth as $G$ then it suffices to (approximately) compute
  $\tw(G)$ and let $H$ be any graph from a well-known class such as
  grids, cliques or expanders with the same treewidth.}.  The notion
of treewidth sparsifiers is closely related to the notion of kernels
for treewidth. A polynomial kernel for treewidth is a map $f$, that,
given an instance $(G,k)$, returns an instance $(G',k')$, with the
property that $\tw(G) \le k$ iff $\tw(G') \le k'$, while ensuring that
the size of the graph $G'$ is polynomial in $k$.  Unless $\NP
\subseteq \coNP/{\sf poly}$ there is no polynomial kernel for
treewidth which follows from the results of Bodlaender et
al.~\cite{tw-no-poly-kernel} and
Drucker~\cite{Drucker12}. Super-linear lower bounds for more general
forms of kernelization are also known \cite{Jansen13}.


Our main result shows that if one is willing to settle for a
poly-logarithmic factor approximation in the treewidth, then there
exist sparsifiers with very strong properties. To state our main
result we need a definition. A graph $H$ is a topological minor of $G$
if $H$ is obtained from $G$ by edge and node deletions, and by
suppressing degree-$2$ nodes\footnote{Note that $H$ is a minor of $G$
  if it can be obtained by edge and node deletions and edge
  contractions. A minor $H$ of a graph $G$ need not be a topological
  minor $G$, however, if the maximum vertex degree in $H$ is at most
  $3$, then $H$ is also a topological minor of $G$.}. Equivalently,
$H$ is a topological minor of $G$ iff a subdivision of $H$ is a
subgraph of $G$. Our main result is summarized in the following
theorem.  \iffull
\begin{theorem}\label{thm: main-topological-minor}
  There is a randomized algorithm, that, given a graph $G$ of
  treewidth at least $k$, with high probability 
  computes a topological minor $H$ of $G$, such
  that:

\begin{itemize}
\item the treewidth of $H$ is  $\Omega(k/\poly\log k)$;  
\item the maximum vertex degree in $H$ is $3$; and
\item $|V(H)|=O(k^4)$.
\end{itemize}
The running time of the algorithm is polynomial in $|V(G)|$ and $k$.
\end{theorem}
\fi

Our result is close to optimal: degree $3$ cannot be reduced, and the
best one can hope for in terms of the size of the sparsifier is
$O(k^2/\poly\log k)$ (when $G$ is a $k \times k$ grid). We also recall that the
best currently known polynomial-time approximation algorithm can only certify
treewidth to within an $O(\sqrt{\log k})$-factor. We conjecture a
strengthening of the theorem to almost optimal parameters.

\begin{conjecture}
  For every graph $G$ with treewidth at least $k$, 
  there \emph{exists} a topological minor $H$ of $G$ 
  such that $\tw(H) = \Omega(k/\poly\log k)$, $|V(H)|=O(k^2)$ and maximum
 vertex degree in $H$ is $3$.
\end{conjecture}

The existence of sparsifiers of size $\poly(k)$ that preserve the
treewidth to within a constant factor remains a very interesting open
question.

\subsection{Treewidth Sparsifiers and Grid Minors}
A fundamental result in Graph Minor Theory is the Grid-Minor Theorem of
Robertson and Seymour \cite{RS-grid}. The theorem states that there is an
integer-valued function $f$, such that any graph $G$ with treewidth at
least $f(g)$ contains a $g \times g$ grid as a minor.  The theorem is
equivalent to showing that $\tw(G) \ge f(g)$ implies that $G$ contains
a {\em wall} of height and width $\Theta(g)$ as a subgraph; see
Figure~\ref{fig:wall}.

\begin{figure}[h]
  \centering
  \includegraphics[height=1in]{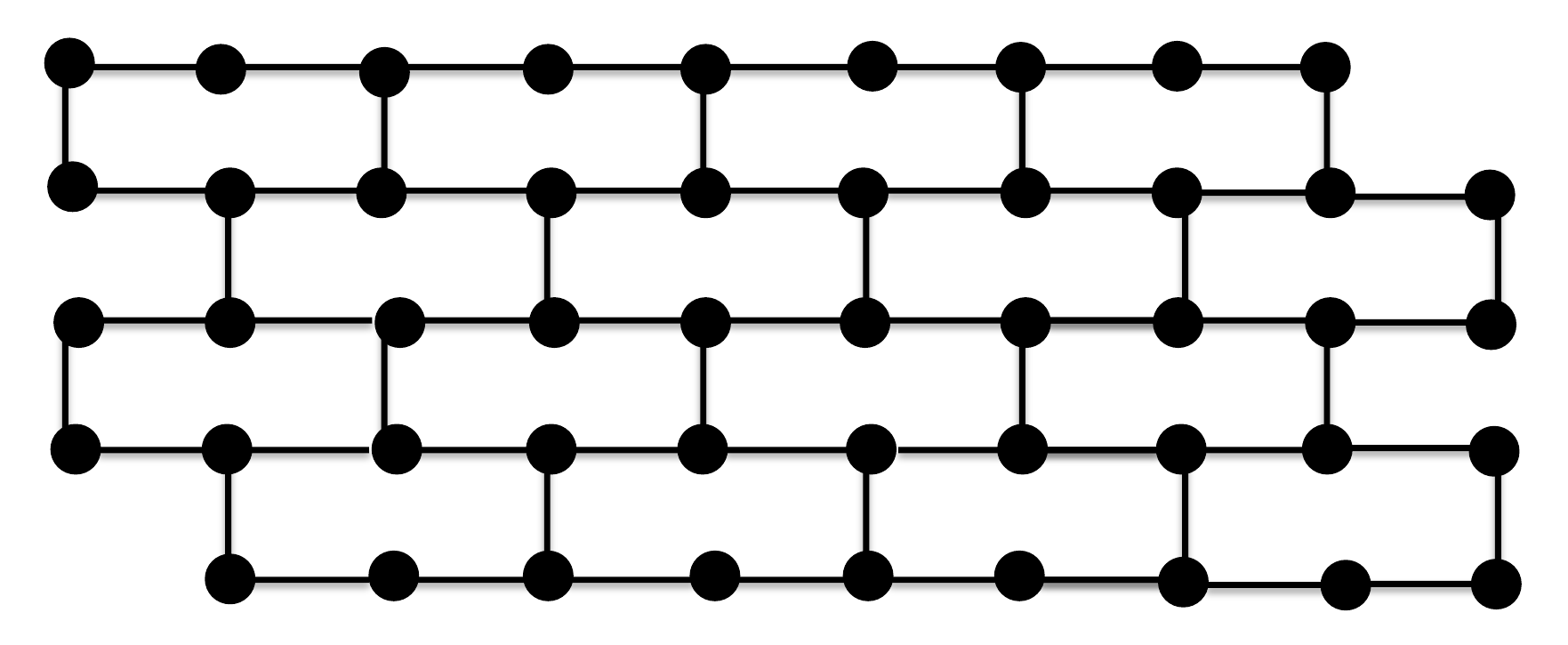}
  \caption{An elementary wall of height and width $5$. A wall is a
    subdivision of an elementary wall.\label{fig:wall}}
\end{figure}

We observe that a wall has maximum vertex degree $3$. Thus, one way to
obtain a degree-$3$ treewidth sparisfier is via the Grid-Minor
Theorem. The original proof of Robertson and Seymour \cite{RS-grid}
showed the existence of $f$ with an iterated exponential dependence on
$g$. Very recently, the first polynomial bound on $f$ was shown in
\cite{CC13-grid}: namely, every graph of treewidth $k$ contains a wall
of size $k^{\delta}$ as a topological minor, where $\delta = 1/98 -
o(1)$. This result implies a degree-$3$ treewidth sparsifier, whose
treewidth is $k^{1/98 - o(1)}$.  In contrast, the sparsifier from
Theorem~\ref{thm: main-topological-minor} has treewidth
$\Omega(k/\polylog k)$.  Moreover, there are graphs with treewidth
$k$, such that the size of the largest wall they contain is
$O(\sqrt{k/\log k})$ \cite{RobertsonST94}. Therefore, one cannot hope
to obtain small sparsifiers that preserve treewidth to within
polylogarithmic factors via the Grid-Minor Theorem. Our construction
bypasses this limitation.

One of our motivations for studying treewidth sparsifiers is improving
the bounds for the Grid-Minor Theorem.  Theorem~\ref{thm:
  main-topological-minor} allows us to focus on subcubic graphs with
the additional property that $|V(G)|$ is polynomial in
$\tw(G)$. Degree-$3$ sparsifiers have particular advantages: in such a
graph, for several applications of interest, one can replace
node-disjoint routing with the easier edge-disjoint routing.
We anticipate that using Theorem~\ref{thm: main-topological-minor} as
a starting point, the bounds on the Grid-Minor Theorem from
\cite{CC13-grid} can be improved.  We also mention
that the fact that $|V(H)|=\poly(k)$ simplifies some
technical parts in the current proof of~\cite{CC13-grid}.

A related application is to the notion of graph immersions (see
\cite{GraphMinors23,Wollan13}).  A graph $G$ admits a strong immersion
of a graph $H$ iff there is an injective mapping $\tau:V(H)
\rightarrow V(G)$ and a mapping $\pi:E(H) \rightarrow \mathcal{P}_G$,
where $\mathcal{P}_G$ is a set of paths in $G$, such that (i) for each
$f=(a,b) \in E(H)$ the path $\pi(f)$ connects $\tau(a)$ and $\tau(b)$;
(ii) for any two edges $f,f' \in E(H)$ the paths $\pi(f)$ and
$\pi(f')$ are edge-disjoint; and (iii) for every $f \in E(H)$ the path
$\pi(f)$ intersects $\tau(V(H))$ only at its endpoints.  Note that $G$
admits $H$ as a topological minor if additionally the paths $\pi(f)$
and $\pi(f')$ are internally {\em node}-disjoint for any
distinct pair $f,f'\in E(H)$. If $G$ is a sub-cubic graph, then
$G$ contains $H$ as a topological minor iff $G$ contains $H$ as a
strong immersion. Therefore, $G$ contains a wall $W$ iff it contains
it as an immersion.  In recent work, Wollan~\cite{Wollan13} defined
the notion of tree-cut width of a graph and showed, using the
Grid-Minor Theorem, that there is a function $g$, such that every
graph with tree-cut width at least $g(r)$ admits an $r$-wall as a weak
immersion. Motivated by this connection, he raised the question of the
existence of degree-$3$ treewidth sparsifiers.  Theorem~\ref{thm:
  main-topological-minor} answers his question (Question 18 in
\cite{Wollan13}) in a near-optimal fashion and we refer the reader to
\cite{Wollan13} for the quantitative and qualitative implications 
to immersions.

Our result can be viewed as providing an approximate kernel for
treewidth, and we hope that it will find applications in preprocessing
graphs for fixed-parameter tractable (FPT) algorithms, and in constructive aspects of
Erdos-P\"{o}sa type theorems.

We now briefly discuss our techniques. We use a combinatorial object,
called a path-of-sets system, that was defined in~\cite{CC13-grid}
(see also Figure~\ref{fig:pos}). Using the construction of the
Path-of-Sets system from~\cite{CC13-grid}, together with the
Cut-Matching Game of Khandekar, Rao and Vazirani~\cite{KRV}, we can
immediately obtain a strong degree-$4$ treewidth sparsifier $H$, with
$\tw(H) = \Omega(k/\polylog k)$. However, the size of $V(H)$ can be 
arbitrarily large. Our main technical contribution
is two-fold. First, we lower the degree of the sparsifier to $3$, by
carefully sub-sampling the edges of $H$. Second, we reduce the size of
the sparsifier to $\poly(k)$. For the second part, we crucially need a
new technical ingredient, that is related to strong vertex-cut
sparsifiers, that we discuss
below. 


\subsection{Sparsifiers Preserving Vertex Cuts}
Suppose we are given any graph $G=(V,E)$ and a pair $S,T\subseteq V$
of vertex subsets, containing $k$ vertices each. We say that the pair
$(S,T)$ is \emph{routable} in $G$ iff there is a set $\pset$ of $k$
disjoint paths connecting the vertices of $S$ to the vertices of $T$
in $G$, and we say that the set $\pset$ of paths \emph{routes} the
pair $(S,T)$. Assume now that we have two pairs of vertex subsets:
$S_1,T_1$, containing $k_1$ vertices each, and $S_2,T_2$ containing
$k_2$ vertices each. We say that both pairs $(S_1,T_1)$, $(S_2,T_2)$
are \emph{separately routable}, or just routable, in $G$ iff there is
a set $\pset$ of paths routing $(S_1,T_1)$, and there is a set $\qset$
paths routing $(S_2,T_2)$ in $G$. Note that a vertex of $G$ may belong
to a path in $\pset$ and a path in 
$\qset$. Our second main result is summarized in the following theorem.

\begin{theorem}\label{thm: 2-flow-main}
  Assume that we are given a graph $G$, two sets $S_1,T_1\subseteq
  V(G)$ of $k_1$ vertices each, and two sets $S_2,T_2\subseteq V(G)$
  of $k_2$ vertices each, such that $k_1\geq k_2$, and the pairs
  $(S_1,T_1)$ and $(S_2,T_2)$ are (separately) routable in $G$. Then
  there are two sets $\pset,\qset$ of paths routing $(S_1,T_1)$ and
  $(S_2,T_2)$ respectively, such that, if $H$ is the graph obtained by
  the union of the paths in $\pset$ and $\qset$, then $\tau(H)\leq
  8k_1^4+ 8k_1$, where $\tau(H)$ is the number of nodes of degree more
  than two in $H$. Moreover, we can find $\pset$ and $\qset$ in time
  polynomial in $n$ and $k_1$.
\end{theorem}

The preceding theorem gives an upper bound on the size of a
topological minor of $G$ that preserves the vertex connectivity
between $S_1,T_1$ and $S_2,T_2$. There are results in the literature
on reduction operations that preserve edge
connectivity~\cite{Lovasz-splitting-off,edge-connectivity} (and also
element connectivity \cite{element-connectivity,ChekuriK09}), however
no such nice operations are available for preserving vertex
connectivity. We briefly discuss some related work on cut sparsifiers
and an open problem on a generalization of Theorem~\ref{thm:
  2-flow-main} that would yield strong sparsifiers that preserve
vertex cuts.

There has been a large amount of work in the recent past on graph
sparsifiers that preserve cuts and flows for subsets of
vertices~\cite{Moitra,LM,CLLM,MM,EGK,vsparsifiers}. We discuss some
closely related work. Given an edge-capacitated graph $G$ and a
terminal set $\tset \subseteq V(G)$, a graph $H$ is a quality-$q$
cut-sparsifier for $\tset$ if (i) $\tset \subseteq V(H)$ and (ii) for
any partition $(A,B)$ of $\tset$, $\mincut_G(A,B) \le \mincut_H(A,B)
\le q \mincut_G(A,B)$ where $\mincut_F(A,B)$ is the minimum edge-cut
separating $A$ from $B$ in a graph $F$. 
 Quality-$1$ sparsifiers
have also been called {\em mimicking networks}  in
prior work~\cite{mimicking0,mimicking1,mimicking2,mimicking3}. 
Leighton and Moitra~\cite{LM} have shown that for any graph $G$, there is a quality-$q$ sparsifier $H$ for $G$ with $q = O(\log k/\log \log k)$ and
$V(H) = \tset$ (here $k = |\tset|$); in other words the sparsifier
does not use any non-terminal (or Steiner) vertices.  There are
instances on which the best quality one can achieve is
$\Omega(\sqrt{\log k})$ if $H$ does not have Steiner
vertices~ \cite{MM}. Even a relatively small number of Steiner vertices can help substantially in improving the
quality of the sparsifier as shown in \cite{vsparsifiers}.  

To simplify the discussion, we restrict our attention to the case where
the terminals in $\tset$ have degree $1$ and all edge capacities are $1$. In
this case constant quality cut-sparsifiers are known with $V(H) =
O(k^3)$ \cite{vsparsifiers,KratschW12}.  The result of Kratsch and
Wahlstr\"{o}m \cite{KratschW12}, in fact, applies in the more general
setting of vertex-cuts, and yields a quality-$1$
sparsifier; we call such a sparsifier a vertex-cut sparsifier to
distinguish it from an edge-cut sparsifier.

However, the sparsifer of \cite{KratschW12} is {\em not} a minor
of the original graph $G$. Sparsifiers that are minors of the original
graph have an advantage that they allow flows (fractional or integral) and minors
in the sparsifier to be transferred back to the original graph $G$ without
any loss. Theorem~\ref{thm: 2-flow-main} gives us a
small-sized minor that preserves the vertex connectivity between two pairs of vertex
subsets. A natural open question is to generalize this result to a larger number of pairs of vertex subsets.

\begin{question}
  Assume that we are given a graph $G$, and $h$ pairs
  of vertex subsets $(S_1,T_1), \ldots, (S_h,T_h)$, such that for
  each $i$: (1) $S_i,T_i \subseteq  V(G)$, (2)
  $|S_i|=|T_i| = k_i \le k$, and (3) $(S_i,T_i)$ are routable in $G$.
  What is the smallest function $f(k,h)$, such that, given any graph $G$ and $(S_1,T_1), \ldots, (S_h,T_h)$ as above, there is always a
  (topological) minor $H$ of $G$ with the property that 
  each $(S_i,T_i)$ is routable in $H$ and $|V(H)| \le f(k,h)$? 
\end{question}

The case when $h = \polylog k$ is of particular interest.  We believe
that a bound on $f(k,h)$ from the preceding question can be
used to obtain a vertex-cut sparsifier $H$ for any graph $G$ and a set $\tset$ of $k$ terminals,
such that $H$ is a minor of $G$, $|V(H)|\leq f(\poly k, h)$ for $h=\poly\log k$, and
the quality of $H$ is $\polylog k$.

\paragraph{Organization}
We prove Theorem~\ref{thm: 2-flow-main} in Section~\ref{sec: 2-flow}.
Section~\ref{sec: PoS system} provides the necessary
background on treewidth and the path-of-sets system.
Theorem~\ref{thm: main-topological-minor} is proved in two steps.
Section~\ref{sec: top minor} gives the proof of a weaker result,
a degree-$4$ sparsifier. 
Section~\ref{sec: degree-3} gives the proof for the degree-$3$ sparsifier.

\label{---------------------------------------------sec: 2-flow-----------------------------------------}
\section{Routing Two Pairs of Vertex Subsets}\label{sec: 2-flow}
In this section we prove Theorem~\ref{thm: 2-flow-main}.  Recall that
a graph $H$ is a \emph{minor} of a graph $G$, iff $H$ can be obtained
from $G$ by a series of edge deletion, vertex deletion, and edge
contraction operations. Equivalently, $H$ is a minor of $G$ iff there
is a map $f:V(H)\rightarrow 2^{V(G)}$ assigning to each vertex $v\in
V(H)$ a subset $f(v)$ of vertices of $G$, such that: (a) for each
$v\in V(H)$, the sub-graph of $G$ induced by $f(v)$ is connected; (b)
if $u,v\in V(H)$ and $u\neq v$, then $f(u)\cap f(v)=\emptyset$; and
(c) for each edge $e=(u,v)\in E(H)$, there is an edge in $E(G)$ with
one endpoint in $f(v)$ and the other endpoint in $f(u)$.  A map $f$
satisfying these conditions is called \emph{a model of $H$ in
  $G$}. Given any subset $X\subseteq V$ of vertices of $G$, we say
that $H$ is an $X$-respecting minor of $G$, iff $X\subseteq
V(H)$. More formally, there is a model $f$ of $H$, where for each
vertex $x\in X$, there is a distinct vertex $v_x\in V(H)$ with
$f(v_x)=\set{x}$. For each $x\in X$, we will usually identify such
vertex $v_x$ with $x$. In particular, every subset $S\subseteq X$ of
vertices of $X$ corresponds to a subset $S'=\set{v_x\mid x\in X}$ of
vertices in $H$, and we will not distinguish between $S$ and $S'$.

Assume that we are given a graph $G$ and two pairs $(S_1',T_1')$,
$(S_2',T_2')$ of vertex subsets, with $|S_1'|=|T_1'|$ and
$|S_2'|=|T_2'|$, that are separately routable in $G$.  We say that a
minor $H$ of $G$ is $(S_1',T_1',S_2',T_2')$-good, iff $H$ is an
$X$-respecting minor for $X=S_1'\cup S_2'\cup T_1'\cup T_2'$, and
$(S_1',T_1')$, $(S_2',T_2')$ are each routable in $H$. We say that it
is $(S_1',T_1',S_2',T_2')$-minimal, iff it is
$(S_1',T_1',S_2',T_2')$-good, and for every edge $e$ of $H$, both the graph
obtained from $H$ by deleting $e$, and the graph obtained from $H$ by contracting $e$, are not
$(S_1',T_1',S_2',T_2')$-good. The main result of this section is the
following theorem.

\begin{theorem}\label{thm: 2-flow-minor}
  Assume that we are given a graph $G$, and sets $S_1',T_1',
  S_2',T_2'\subseteq V(G)$ of $k$ vertices each, such that the pairs
  $(S_1',T_1')$ and $(S_2',T_2')$ are (separately) routable in
  $G$. Assume further that vertices in $S_1',T_1',S_2',T_2'$ are
  distinct, and have degree $1$ each in $G$. Let $H$ be any
  $(S_1',T_1',S_2',T_2')$-minimal minor of $G$. Then $|V(H)|\leq
  4k^4+4k$.
\end{theorem}

\iffull
We start by showing that Theorem~\ref{thm: 2-flow-main} follows from
Theorem~\ref{thm: 2-flow-minor}. Let $G$ be the input graph, and
$(S_1,T_1),(S_2,T_2)$ the given pairs of vertex subsets. We denote
$k_1=k$, and add $\Delta=k_1-k_2$ new edges
$e_1=(a_1,b_1),\ldots,e_{\Delta}=(a_{\Delta},b_{\Delta})$, whose
endpoints are distinct, to the graph. The vertices
$\set{a_1,\ldots,a_{\Delta}}$ are then added to $S_2$, and the
vertices $b_1,\ldots,b_{\Delta}$ are added to $T_2$, so
$|S_2|=|T_2|=|S_1|=|T_1|=k$. The new graph then contains a set of
paths routing $(S_1,T_1)$, and a set of paths routing $(S_2,T_2)$.

We add a new set $S_1'$ of $k$ vertices to the graph, and connect each
vertex in $S_1'$ to a distinct vertex in $S_1$ with an edge. We
construct sets $S_2',T_1',T_2'$ of vertices and connect them to the
vertices in $S_2,T_1,T_2$, respectively, in a similar manner. Let $G'$
be this final graph. Then $G'$ contains a set of paths routing
$(S_1',T_1')$, and a set of paths routing $(S_2',T_2')$. The vertices
in $S_1',T_1',S_2',T_2'$ are distinct, and have degree $1$ each in
$G'$. We now compute any $(S_1',T_1',S_2',T_2')$-minimal minor $H$ of
$G'$.  Let $f : V(H) \rightarrow 2^{V(G')}$ be a map to cerfity that $H$
is a minor of $G'$. Let $\pset'$ be the set of paths routing
$(S_1',T_1')$, and $\qset'$ a set of paths routing $(S_2',T_2')$ in
$H$. We use the sets of paths $\pset',\qset'$ to define the sets
$\pset,\qset$ of paths routing $(S_1,T_1)$ and $(S_2,T_2)$,
respectively, in $G$. This mapping is the natural one; we extend a
path $P' \in \pset' \cup \qset'$ in $H$ to a path $P$ in $G$ by replacing each
vertex $v \in P'$ by a path contained in $G[f(v)]$ (the connected
sub-graph of $G$ corresponding to $v$) that connects the two edges
$e,e'$ of $P'$ incident on $v$. Since only two paths from $\pset' \cup
\qset'$ can contain a node $v \in V(H)$, it is not hard to find paths
through $f(v)$ for them with at most $2$ nodes of degree $3$ or more
in $f(v)$; this will ensure that the number of vertices whose degree
is more than $2$ in the resulting graph is at most twice their number
in $H$.  The formal argument is given below.

Consider any path $P'\in \pset'$, and assume
that $P'=(s=v_0,v_1,v_2,\ldots,v_r,v_{r+1}=t)$, so $s\in S_1',t\in
T_1'$. For each $0\leq i\leq r$, we denote the edge $(v_i,v_{i+1})$ by
$e_i^{P'}$. The new path $P$ contains the edges
$e_1^{P'},\ldots,e_{r-1}^{P'}$. Additionally, for each vertex $v_i$,
for $1\leq i\leq r$, it contains an arbitrary path $R_{v_i}(P)$,
connecting the endpoints of $e_{i-1}^{P'}$ and $e_{i}^{P'}$ that are
contained in $f(v_i)$, such that $R_{v_i}(P)\subseteq G[f(v_i)]$. Notice
that since $G[f(v_i)]$ is connected, such a path exist. Let
$\pset=\set{P\mid P'\in \pset'}$ be the resulting set of paths. Since
the paths in $\pset'$ are node-disjoint, so are the paths in
$\pset$. It is then immediate to see that $\pset$ routes $(S_1,T_1)$
in $G$.

Consider now some path $Q'\in \qset'$, and assume that
$Q'=(s=v_0,v_1,v_2,\ldots,v_r,v_{r+1}=t)$, so $s\in S_2',t\in
T_2'$. If $v_1=a_j$ for some $1\leq j\leq \Delta$, then $v_2=b_j$ must
hold, since $e_j$ is the only edge incident on $a_j$ in $G'$, in
addition to $(s,a_j)$. Therefore, $r=2$, and $Q'=(s,a_j,b_j,t)$. We
discard $Q'$ from $\qset'$. Otherwise, $Q'$ cannot contain any
vertices in $\set{a_1,b_1,\ldots,a_{\Delta},b_{\Delta}}$.  For each
$0\leq i\leq r$, we denote the edge $(v_i,v_{i+1})$ by $e_i^{Q'}$. The
new path $Q$ contains the edges
$e_1^{Q'},\ldots,e_{r-1}^{Q'}$. Additionally, for each vertex $v_i$,
for $1\leq i\leq r$, it contains some path $R_{v_i}(Q)$, connecting
the endpoints of $e_{i-1}^{Q'}$ and $e_{i}^{Q'}$ that are contained in
$f(v_i)$, such that $R_{v_i}(Q)\subseteq G[f(v_i)]$. The path
$R_{v_i}(Q)$ is constructed as follows. If $v_i$ does not belong to
any path in $\pset'$, then $R_{v_i}(Q)$ is any path connecting the
endpoints of $e_{i-1}^{Q'}$ and $e_{i}^{Q'}$ that are contained in
$f(v_i)$, such that $R_{v_i}(Q)\subseteq G[f(v_i)]$. Otherwise, let
$P'\in \pset'$ be the path containing $v$. Let $R_1$ be the
intersection of the corresponding path $P\in \pset$ with $G[f(v_i)]$
(which must be a path), and let $R_2$ be any path contained in
$G[f(v_i)]$, that connects the endpoints of $e_{i-1}^{Q'}$ and
$e_{i}^{Q'}$ that belong to $f(v_i)$. If $R_1$ and $R_2$ are disjoint,
then we let $R_{v_i}(Q)=R_2$. Otherwise, let $u$ be the first vertex
on $R_2$ that belongs to $R_1$, and let $v$ be the last vertex on
$R_2$ that belongs to $R_1$. Let $R_2'\subseteq R_2$ be the segment of
$R_2$ from its beginning until the vertex $u$, and $R_2''\subseteq
R_2$ the segment of $R_2$ from $v$ to its end. Let $R_1'\subseteq R_1$
be the segment of $R_1$ between $u$ and $v$. We then let $R_{v_i}(Q)$
be the concatenation of $R_2',R_1'$ and $R_2''$. Notice that in the
graph obtained by the union of $R_1$ and $R_{v_i}(Q)$, there are at
most two vertices whose degree is more than $2$ --- the vertices $u$ and
$v$. For each vertex $z$ that serves as an endpoint of the paths $R_1$
and $R_{v_i}(Q)$, if $z\neq u,v$, then the degree of $z$ is $1$ in
this graph.
Let $\qset$ be the final set of paths obtained after processing all
the paths in $\qset'$. Then it is immediate to see that $\qset$ routes
the original pair $(S_2,T_2)$ of vertex subsets in $G$. Let $H'$ be the graph obtained from the union of all paths in $\pset\cup \qset$. Then for each $v\in V(H)$, $f(v)\cap V(H')$
contains at most two vertices whose degree in $H'$ is more than $2$, so
$\tau(H')\leq 2|V(H)|\leq 8k_1^4+8k_1$.  This completes the proof of
Theorem~\ref{thm: 2-flow-main}.
\fi

In the rest of this section, we focus on the proof of
Theorem~\ref{thm: 2-flow-minor}. For simplicity, we denote
$S_1',S_2',T_1',T_2'$ by $S_1,S_2,T_1$ and $T_2$, respectively.  Let
$H$ be a $(S_1,T_1,S_2,T_2)$-minimal minor of $G$. Let $\rset$ be a
set of paths routing $(S_1,T_1)$ in $H$. We will often refer to the
paths in $\rset$ as \emph{red paths}, and we will think of these paths
as directed from $S_1$ towards $T_1$ (even though in general the graph
is undirected). Similarly, let $\bset$ be the set of paths routing
$(S_2,T_2)$ in $H$. We refer to the paths in $\bset$ as \emph{blue
  paths}, and view them as directed from $S_2$ to $T_2$. Notice that a
vertex in $S_1\cup T_1$ cannot participate in a blue path, since its
degree is $1$, and all vertices in $S_1,T_1,S_2,T_2$ are
distinct. Similarly, a vertex in $S_2\cup T_2$ cannot participate in a
red path. An edge $e\in E(H)$ may belong to a red path, or to a blue
path, but not both, since otherwise we could contract $e$ and obtain a minor
that is still $(S_1,T_1,S_2,T_2)$-good, contradicting the minimality
of $H$; this is possible since $e$ is not incident on $S_1\cup S_2\cup
T_1\cup T_2$. The edges that belong to the paths in $\rset$ are
called \emph{red edges}, and the edges that belong to the paths in
$\bset$ are called \emph{blue edges}. From the minimality of $H$,
every edge is either red or blue. We will refer to the vertices in
$S_1\cup S_2\cup T_1\cup T_2$ as the terminals of $H$. From the
minimality of $H$, every non-terminal vertex belongs to one red path
and one blue path, and is incident on exactly two red edges and
exactly two blue edges.  Assume now that there is another set
$\rset'\neq \rset$ of paths in $H$ routing $(S_1,T_1)$. Then there
must be some red edge in $H$ that does not belong to any path of $\rset'$,
contradicting the minimality of $H$. Therefore, $\rset$ is the unique
set of paths routing $(S_1,T_1)$ in $H$, and similarly, $\bset$ is the
unique set of paths routing $(S_2,T_2)$ in $H$.  We prove the following
theorem.

\begin{theorem}\label{thm: labeling}
  We can efficiently compute an assignment of labels in
  $L=\set{\ell_1,\ell_2,\ldots,\ell_{2k}}$ to the vertices of $V(H)$,
  such that each vertex in $V(H)$ is assigned one label, and for every
  pair $R\in \rset$, $B\in \bset$ of paths, if two vertices $v$ and
  $v'$ belong to both $R$ and $B$, and are assigned the same label,
  then they appear in the same order on $R$ and on $B$.
\end{theorem}

Before we prove Theorem~\ref{thm: labeling}, let us first complete the
proof of Theorem~\ref{thm: 2-flow-minor} assuming it. Let $\ell:
V(H)\rightarrow L$ be the labeling computed by Theorem~\ref{thm:
  labeling}. Next, we switch $S_1$ and $T_1$, so that the directions
of the paths in $\rset$ are reversed. We apply Theorem~\ref{thm:
  labeling} again to this new setting, and obtain another labeling
$\ell': V(H)\rightarrow L'$, where
$L'=\set{\ell_1',\ell_2',\ldots,\ell_{2k}'}$.

Assume for contradiction that $|V(H)|\geq 4k^4+4k+1$. Every
non-terminal vertex $v$ can be associated with a quadruple
$(R,B,\ell_i,\ell'_j)$, where $R$ and $B$ are the red and the blue
paths on which $v$ lies, $\ell_i$ is the label assigned to $v$ by
$\ell$, and $\ell'_j$ is the label assigned to $v$ by $\ell'$. Since
the total number of such quadruples is $4k^4$, there is a pair $u,v$
of non-terminal vertices that have the same quadruple
$(R,B,\ell_i,\ell'_j)$. As $u$ and $v$ are assigned the same label by
$\ell$, they must appear in the same order on $R$ and $B$. Assume
w.l.o.g. that $u$ appears before $v$ on both these paths. However,
since both these vertices are assigned the same label by $\ell'$, and
since the red paths were reversed when computing $\ell'$, the order of
$u$ and $v$ on paths $R$ and $B$ must be reversed, a contradiction.
In order to complete the proof of Theorem~\ref{thm: 2-flow-minor}, it
now only remains to prove Theorem~\ref{thm: labeling}.

\subsection*{Proof of Theorem~\ref{thm: labeling}}
Let $\tH$ be the directed counterpart of the graph $H$, where we
direct all red edges along the direction of the red paths from $S_1$
to $T_1$, and we direct the blue edges similarly along the blue paths
from $S_2$ to $T_2$. The main combinatorial object that we use in the
proof is a chain. A chain $Z$ is a directed (not necessarily simple)
path in graph $\tH$, such that the edges of $Z$ are alternating red
and blue edges. In other words, if the edges of $Z$ are
$e_1,e_2,\ldots,e_r$ in this order, then all odd-indexed edges
are red and all even-indexed edges are blue, or vice versa. The rest
of the proof consists of three steps. First, we show that every chain
must be a simple path, so no vertex may appear twice on a chain. If
this is not the case, we will show that $\rset$ is not a unique set of
paths routing $(S_1,T_1)$, or that $\bset$ is not a unique set of
paths routing $(S_2,T_2)$, leading to a
contradiction. In the second step, we construct a collection of $2k$
chains using a natural greedy algorithm: start from some source, and
then follow alternatively red and blue edges, while possible. We will
show that every vertex of $H$ belongs to at least one chain (but may
belong to more than one). We then associate a separate label with each
chain, and assign all vertices that belong to a chain the same
label. If a vertex belongs to several chains, then one of the
corresponding labels is assigned arbitrarily. Finally, we prove that
for every path $P\in \rset\cup \bset$ and every chain $Z$, if $v$ and
$v'$ are two vertices that belong to both $P$ and $Z$, then they must
appear in the same order on $P$ and on $Z$.

Before we proceed, we define two auxiliary structures: red and blue
cycles. Let $C$ be a directed simple cycle in the graph $\tH$ (so
every vertex may appear at most once on $C$). We say that it is a
\emph{blue cycle} iff we can partition $C$ into an even number of
edge-disjoint consecutive segments
$\sigma_1,\sigma_2,\ldots,\sigma_{2r}$, where $r>0$; for all $1\leq
i\leq r$, $\sigma_{2i}$ consists of a single red edge, and
$\sigma_{2i-1}$ is a non-empty path that only consists of blue
edges. Every edge of $C$ belongs to exactly one segment, and every
consecutive pair of segments shares one vertex (if $r=1$ then the two
segments share two vertices --- the endpoints of the segments). A red
cycle is defined similarly, with the roles of the red and the blue
segments reversed. We start by showing that $\tH$ cannot contain a red
or a blue cycle. 

\begin{lemma}
  \label{lem:no-red-blue-cycle}
  Graph $\tH$ cannot contain a red cycle or a blue cycle.
\end{lemma}

\iffull
\begin{proof}
  We prove for blue cycles; the proof for red cycles is similar. Let
  $C$ be a blue cycle in $\tH$, and let
  $\sigma_1,\sigma_2,\ldots,\sigma_{2r}$ be the corresponding segments
  of $C$. Let $H'$ be the graph obtained from $H$ by deleting all
  edges participating in the segments $\sigma_{2i-1}$, for $1\leq
  i\leq r$ (that is, the blue segments). We claim that both
  $(S_1,T_1)$ and $(S_2,T_2)$ remain routable in $H'$, contradicting
  the minimality of $H$. Since we only deleted blue edges, it is clear
  that $(S_1,T_1)$ remains routable via the paths in $\rset$. We now
  show that $H'$ contains a collection of paths routing $(S_2,T_2)$.

  Let $A$ denote the set of all vertices $a$, such that $a$ is the
  last vertex of some blue segment $\sigma_{2i-1}$ of $C$, for $1\leq
  i\leq r$, and let $B$ denote the set of all vertices $b$, such that
  $b$ is the first vertex of some blue segment $\sigma_{2i-1}$ of
  $C$. Then $|A|=|B|=r$, and the red edges of $C$ define a complete
  matching between $A$ and $B$. Let $\Sigma$ be the collection of
  paths obtained from $\bset$, by deleting all blue edges that
  participate in the cycle $C$ (we do not include $0$-length paths in
  $\Sigma$). Then $\Sigma$ is a collection of disjoint paths, that
  only contain blue edges, which route the pair $(S_2\cup A)$ and
  $(T_2\cup B)$.  Notice that every path in $\Sigma$ contains at
  least two vertices: this is since the terminals cannot belong to $C$
  as their degrees are $1$, and every vertex appears on $C$ at most
  once. Therefore, all vertices in $S_1,S_2,A,B,T_1,T_2$ are distinct.

  We now construct the following directed graph $F$: the vertices of
  $F$ are $S_2\cup A\cup T_2\cup B$. There is a directed edge $(u,v)$,
  for $u\neq v$, in $F$ iff there is a (directed) red edge $(v,u)$ in
  $C$ (in which case we say that $(u,v)$ is a red edge), or there is
  some path in $\Sigma$ that starts at $u$ and terminates at $v$ (in
  which case we say that $(u,v)$ is a blue edge). Notice that in the
  graph $F$, every vertex in $S_2$ has one outgoing edge and no other
  incident edges; every vertex in $T_2$ has one incoming edge and no
  other incident edges; every vertex in $A$ has one outgoing edge
  (blue), and one incoming edge (red); and every vertex in $B$ has one
  incoming edge (blue) and one outgoing edge (red). It is then easy to
  see that $F$ contains $k$ directed disjoint paths connecting $S_2$
  to $T_2$, and this gives a set of paths routing $(S_2,T_2)$  in $H'$,
  contradicting the minimality of $H$.
\end{proof}
\fi


The claim below essentially follows from the preceding lemma. 

\begin{claim}\label{claim: no cycles on chain}
If $Z$ is a chain, then every vertex of $V(H)$ may appear on $Z$ at most once.
\end{claim}

\iffull
\begin{proof}
  Assume otherwise, and let $Z$ be any chain, such that some vertex
  appears more than once on $Z$. Then there is a segment $Z'$ of $Z$,
  such that both endpoints of $Z'$ are the same vertex $v$, but every
  other vertex appears at most once on $Z'$, and $v$ is not an inner
  vertex of $Z'$. Notice that $Z'$ must contain at least two edges.
  Then $Z'$ defines a simple directed cycle in $\tH$. Moreover, if 
  both edges incident on $v$ in $Z'$ are blue, then $Z'$ is a blue
  cycle; if both edges are red then $Z'$ is a red cycle; otherwise it
  is both a red and a blue cycle. Since $\tH$ cannot contain a red or
  a blue cycle, every vertex appears on $Z$ at most once.
\end{proof}
\fi

We define a collection $\zset$ of $2k$ chains in $\tH$,
and prove that every vertex of $\tH$ belongs to at least one chain.
Let $s\in S_1\cup S_2$, and let $e$ be the unique edge leaving $s$. We
start building the chain by adding $e$ to the chain. If the last edge
added to the chain $e'=(u,v)$ is a red edge, and there is a blue edge
leaving $v$ in $\tH$, then we add the unique blue edge leaving $v$ in
$\tH$ to the chain; if no such edge exists, we complete the
construction of the chain --- in this case, $v\in T_1\cup T_2$ must
hold. Similarly, if the last edge added to the chain $e'=(u,v)$ is a
blue edge, and there is a red edge leaving $v$ in $\tH$, then we add
the unique red edge leaving $v$ in $\tH$ to the chain; if no such edge
exists, we complete the construction of the chain. Overall, we
construct one chain starting from each vertex in $S_1\cup S_2$,
obtaining $2k$ chains. Let $\zset$ denote the resulting collection of
the chains. 

\begin{claim}
  \label{claim: chain covers all vertices}
  Every vertex of $\tH$ belongs to at least one chain.
\end{claim}
\iffull
\begin{proof}
  We prove a slightly stronger claim: that every edge of $\tH$ belongs
  to at least one chain.

Let $e$ be any edge of $\tH$, and assume w.l.o.g. that it is a blue
edge. We construct a chain $P$ from the end to the beginning, and we
start by adding the edge $e$ to $P$ as the last edge of $P$. Assume
that the last edge added to $P$ was $e'=(u,v)$. If $u\in S_1\cup S_2$,
then we terminate the construction of $P$; otherwise, there must be a
red edge entering $u$ and a blue edge entering $u$. If $e'$ is a red
edge, then we add the unique blue edge entering $u$ to $P$, and
otherwise we add the unique red edge entering $u$ to $P$, and continue
to the next iteration. Since at every step, the current path $P$ is a
valid chain, and no vertex may appear twice on a chain, this process
will eventually stop at some vertex $s\in S_1\cup S_2$. Then the
unique chain $Z\in \zset$ that starts from $s$ must contain $P$ as a
sub-path, and hence must contain the edge $e$.
\end{proof}
\fi

Our final step is the following claim.

\begin{claim}\label{claim: order}
  Let $Z$ be a chain, and assume that it contains two vertices
  $v,v'\in V(P)$, where $P\in \rset\cup \bset$. Assume further that
  $v$ appears before $v'$ on $Z$. Then $v$ appears before $v'$ on $P$.
\end{claim}

\iffull
\begin{proof}
  Assume otherwise. Then there must be two vertices $u,u'$ that appear
  on both $Z$ and $P$, such that no other vertex of $P$ appears
  between $u$ and $u'$ on $Z$, $u$ appears before $u'$ on $Z$, and it
  appears after $u'$ on $P$. Indeed, consider the segment $Z^*$ of $Z$
  between $v$ and $v'$. If this segment contains no other vertex of
  $P$, then we are done. Otherwise, assume w.l.o.g. that $v$ appears
  before $v'$ on $Z$, and let $v_0=v,v_1,\ldots,v_x=v'$ be the
  vertices of $P\cap Z$, that appear on $Z^*$ in this order. Since $v$
  appears after $v'$ on $P$, there must be a consecutive pair
  $v_i,v_{i+1}$ of vertices, such that $v_i$ appears after $v_{i+1}$
  on $P$. We then set $u=v_i$ and $u'=v_{i+1}$.

  Let $Z'$ be the segment of the chain $Z$ between $u$ and $u'$, and
  let $P'$ be the segment of $P$ between $u'$ and $u$. Observe that
  $P'\cap Z'=\set{u',u}$, and $Z'$ contains at least one edge whose
  color is opposite from the color of $P$. If $P$ is a blue path, then
  $Z'\cup P'$ is a blue cycle; otherwise it is a red cycle, a
  contradiction.
\end{proof}
\fi

We are now ready to assign labels to the vertices of $H$. Let
$\zset=\set{C_1,C_2,\ldots,C_{2k}}$. Fix any vertex $v\in H$, and let
$C_i\in \zset$ be any chain that contains $v$. We then assign to $v$
the label $\ell_i$.

Consider now any pair $R\in \rset$, $B\in \bset$ of paths, and let
$v,v'$ be two vertices that have the same label $\ell_i$ and appear on
both $R$ and $B$. Assume w.l.o.g. that $v$ appears before $v'$ on
chain $C_i$. Then from Claim~\ref{claim: order}, $v$ must appear
before $v'$ on both $R$ and $B$.  This completes the proof of
Theorem~\ref{thm: labeling}, and hence of Theorem~\ref{thm:
  2-flow-minor} and Theorem~\ref{thm: 2-flow-main}.

 \label{---------------------------------------------sec: PoS system-minor-----------------------------------------}
\section{Background on Treewidth and Path-of-Sets System}\label{sec: PoS system}

\iffull In this section we define some graph-theoretic notions and
summarize some previous results that we use in the proof of
Theorem~\ref{thm: main-topological-minor}. We also define a
combinatorial object that plays a central role in the proof --- the
path-of-sets system from \cite{CC13-grid}.

Given a graph $G=(V,E)$ and a set $A\subseteq V$ of vertices, we
denote by $E_G(A)$ the set of edges with both endpoints in $A$, and by
$\out_G(A)$ the set of edges with exactly one endpoint in $A$.  For
disjoint sets of vertices $A$ and $B$, the set of edges with one end
point in $A$ and the other in $B$ is denoted by $E_G(A,B)$. For a
vertex $v$ in a graph $G$ we use $d_G(v)$ to denote its degree. We may
omit the subscript $G$ if it is clear from the context. Given a set
$\pset$ of paths in $G$, we denote by $V(\pset)$ the set of all
vertices participating in paths in $\pset$, and similarly, $E(\pset)$
is the set of all edges that participate in paths in $\pset$. We
sometimes refer to sets of vertices as \emph{clusters}. A path $P$ in
a graph $G$ is a \emph{$2$-path} iff every inner vertex $v$ in $P$ has
$d_G(v)=2$. It is a \emph{maximal $2$-path} iff the degrees of the
endpoints of $P$ are both different from $2$. Given a set $\pset$ of
paths, we denote by $J(\pset)$ the graph obtained by the union of all
the paths in $\pset$. Given a graph $H$, let $\tau(H)$ denote the number of vertices of $H$ whose degree is more than $2$ in $H$.

We now define the notion of linkedness and the different notions of
well-linkedness that we use.

\begin{definition}
  We say that a set $\tset$ of vertices is
  $\alpha$-well-linked\footnote{This notion of well-linkedness is
    based on edge-cuts and we distinguish it from node-well-linkedness
    that is directly related to treewidth. For technical reasons it is
    easier to work with edge-cuts and hence we use the term well-linked to mean
    edge-well-linkedness, and explicitly use the term
    node-well-linkedness when necessary.} in $G$, iff for any
  partition $(A,B)$ of the vertices of $G$ into two subsets,
  $|E(A,B)|\geq \alpha\cdot \min\set{|A\cap \tset|,|B\cap \tset|}$.
\end{definition}

\begin{definition}
  We say that a set $\tset$ of vertices is \emph{node-well-linked} in
  $G$, iff for any pair $(\tset_1,\tset_2)$ of equal-sized subsets of
  $\tset$, there is a collection $\pset$ of $|\tset_1|$ {\bf
    node-disjoint} paths, connecting the vertices of $\tset_1$ to the
  vertices of $\tset_2$. (Note that $\tset_1$, $\tset_2$ are not
  necessarily disjoint, and we allow empty paths).
\end{definition}

The two different notions of well-linkedness are closely related.  In
particular, suppose $\tset$ is $\alpha$-well-linked in a graph $G$ of
maximum degree $\Delta$. Then there is a large subset $\tset'\subseteq
\tset$ of vertices that is node-well-linked in $G$, as shown in the
following theorem.

 \begin{theorem}[Theorem 2.2 in~\cite{CC13-grid}]
   \label{thm: grouping}
   Suppose we are given a connected graph $G=(V,E)$ with maximum
   vertex degree $\Delta$, and a subset $\tset$ of $\kappa$ vertices
   called terminals, such that $\tset$ is $\alpha$-well-linked in $G$,
   for some $\alpha<1$. Then there
   is a subset $\tset'\subset \tset$ of
   $\Omega\left(\frac{\alpha\kappa }{\Delta}\right )$ terminals, such
   that $\tset'$ is node-well-linked in $G$.
\end{theorem}

The following well-known lemma summarizes an important connection
between treewidth and node-well-linkedness.
\begin{lemma}[\cite{Reed-chapter}]
  \label{lem:tw-wl}
  Let $k$ be the size of the largest node-well-linked set in $G$. Then
  $k \le \tw(G) \le 4k$.
\end{lemma}

Combining Theorem~\ref{thm: grouping} with Lemma~\ref{lem:tw-wl}, we
obtain the following theorem.

\begin{theorem}\label{thm: weak well-linkedness to tw}
  Let $G$ be any graph with maximum vertex degree $\Delta$, and
  $\tset$ a subset of $\kappa$ vertices, such that $\tset$ is
  $\alpha$-well-linked in $G$, for $\alpha< 1$. Then the treewidth of
  $G$ is $\Omega(\alpha\kappa/\Delta)$.
\end{theorem}

A notion closely related to well-linkedness is that of linkedness,
where we require good connectivity between a pair of disjoint vertex subsets.

\begin{definition}
  We say that two disjoint vertex subsets $A$ and $B$ are \emph{linked} in
  $G$ iff for any pair of equal-sized subsets $A'\subseteq A$,
  $B'\subseteq B$ there is a set $\pset$ of $|A'|$ node-disjoint
  paths connecting $A'$ to $B'$ in $G$.
\end{definition}

\paragraph{Path-of-Sets System}
A central combinatorial object that we use in the proof of
Theorems~\ref{thm: main-topological-minor} is a path-of-sets system,
that was introduced in~\cite{CC13-grid} (a somewhat similar object,
called a grill, was introduced by Leaf and
Seymour~\cite{LeafS12}). See Figure~\ref{fig:pos}.

\begin{definition}
  A path-of-sets system $(\sset,\bigcup_{i=1}^{r-1}\pset_i)$ of width
  $r$ and height $h$ consists of:

\begin{itemize}
\item A sequence $\sset=(S_1,\ldots,S_r)$ of $r$ disjoint vertex
  subsets of $G$, where for each $i$, $G[S_i]$ is connected;

\item For each $1\leq i\leq r$, two disjoint sets $A_i,B_i\subseteq
  S_i$ of $h$ vertices each, such that $A_i$ and $B_i$ are linked in
  $G[S_i]$;
\item For each $1\leq i<r$, a set $\pset_i$ of $h$ disjoint paths,
  routing $(B_i,A_{i+1})$, such that all paths in $\bigcup_i\pset_i$
  are mutually disjoint, and do not contain the vertices of
  $\bigcup_{S_i\in \sset}S_i$ as inner vertices,
\end{itemize}

We say that it is a \emph{strong} path-of-sets system, if additionally
for each $1\leq i\leq r$, $A_i$ is node-well-linked in $G[S_i]$, and
the same holds for $B_i$.
\end{definition}

\begin{figure}[htb]
  \centering
  \includegraphics[width=6in]{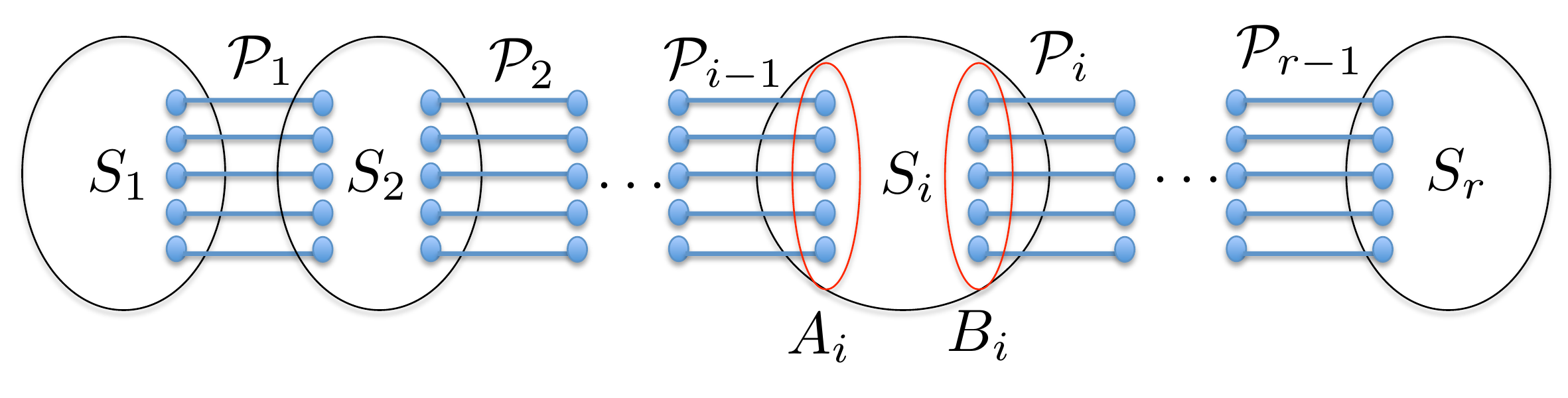}
  \caption{Path-of-Sets System}
  \label{fig:pos}
\end{figure}

The following theorem, that was proved in~\cite{CC13-grid}, is the
starting point of the proof of Theorem~\ref{thm:
  main-topological-minor}.

\begin{theorem}[Theorem 3.2 in~\cite{CC13-grid}] 
\label{thm: strong PoS system}
Let $G$ be any graph of treewidth $k$, and let $h,r>1$ be integral
parameters, such that for some large enough constants $c$ and $c'$,
$k/\log^{c'}k>chr^{48}$. Then there is an efficient randomized
algorithm, that, given $G,h$ and $r$, w.h.p. computes a strong
path-of-sets system of height $h$ and width $r$ in $G$.
\end{theorem}

{\bf Expanders and the Cut-Matching Game.}
We say that a (multi)-graph $G=(V,E)$ is an $\alpha$-expander, iff
$\min_{\stackrel{S\sse V:}{|S|\leq
    |V|/2}}\set{\frac{|E(S,\nots)|}{|S|}}\geq \alpha$.  We use the
cut-matching game of Khandekar, Rao and Vazirani~\cite{KRV} to
construct an expander that can be appropriately embedded in a
graph. In this game, we are given a set $V$ of $N$ vertices, where $N$
is even, and two players: a cut player, whose goal is to construct an
expander $X$ on the set $V$ of vertices, and a matching player, whose
goal is to delay its construction. The game is played in
iterations. We start with the graph $X$ containing the set $V$ of
vertices, and no edges.  In each iteration $j$, the cut player
computes a bi-partition $(A_j,B_j)$ of $V$ into two equal-sized sets,
and the matching player returns some perfect matching $M_j$ between
the two sets. The edges of $M_j$ are then added to $X$. Khandekar, Rao
and Vazirani have shown that there is a strategy for the cut player,
guaranteeing that after $O(\log^2N)$ iterations, no matter the
strategy of the matching player, the resulting graph is 
a $\half$-expander w.h.p. Subsequently, Orecchia et
al.~\cite{better-CMG} have shown the following improved bound:


\begin{theorem}[\cite{better-CMG}]\label{thm: CMG}
  There is a probabilistic algorithm for the cut player, such that, no
  matter how the matching player plays, after $\gkrv(N)=O(\log^2N)$
  iterations, graph $X$ is an $\alphaCMG(N)=\Omega(\log N)$-expander,
  with constant probability.
\end{theorem}

Our algorithms work by embedding an expander $X$ into a sub-graph of
$G$. The embedding of the expander is then used to certify the
treewidth. We use the following notion of embedding.

\begin{definition}
  Let $G,X$ be graphs. An embedding $\phi$ of $X$ into $G$ maps every
  vertex $v\in X$ to a connected subgraph $C_v\subseteq G$, and every
  edge $e=(u,v)\in E(X)$ to a path $P_e$ in graph $G$, whose endpoints
  belong to $C_v$ and $C_u$, respectively. We say that the
  \emph{congestion} of the embedding is at most $c$, iff every edge of
  $G$ belongs to at most $c-1$ paths in $\set{P_e\mid e\in E(X)}$ and
  at most one graph $\set{C_v\mid v\in V(X)}$.
\end{definition}

In the next simple claim, we show that if we can embed a
$\kappa$-vertex expander with congestion at most $c$ into a graph $H$
with bounded vertex degree, then the treewidth of $H$ is large.

\begin{claim}\label{claim: embed an expander treewidth}
  Let $X$ be an $\alpha$-expander on $\kappa$ vertices for $\alpha<1$,
  with maximum vertex degree $\Delta'$, and let $H$ be a graph with
  maximum vertex degree at most $\Delta$, such that that there is an
  embedding of $X$ into $H$ with congestion $\eta$. Then $\tw(H) = 
  \Omega\left (\frac{\alpha\kappa}{\eta\Delta\Delta'}\right )$.
\end{claim}

\begin{proof}
  For each vertex $v\in V(X)$, let $t_v$ be an arbitrary vertex in
  $C_v$, and let $\tset=\set{t_v\mid v\in V(X)}$. Since $|X|=\kappa$,
  from Theorem~\ref{thm: weak well-linkedness to tw}, it is enough to
  show that $\tset$ is $\frac{\alpha}{2\eta\Delta'}$-well-linked in
  $H$. Let $(A,B)$ be any partition of $V(H)$, denote $\tset_A=A\cap
  \tset$, $\tset_B=B\cap \tset$, and $E'=E(A,B)$. Assume w.l.o.g. that
  $|\tset_A|\leq |\tset_B|$. Then it is enough to show that $|E'|\geq
  |\tset_A|\cdot \frac{\alpha}{2\eta\Delta'}$.

  We partition $\tset_A$ into two subsets, $\tset_A'$ and $\tset_A''$,
  as follows. For each vertex $t_v\in \tset_A$, if $C_v\subseteq A$,
  then we add $t_v$ to $\tset_A''$, and otherwise we add it to
  $\tset_A'$. We partition $\tset_B$ into two subsets, $\tset_B'$ and
  $\tset_B''$ similarly.  Let $\kappa'=|\tset_A|$. Assume first that
  $|\tset_A'|\geq \kappa'/2$. Then for each vertex $t_v\in \tset_A'$,
  at least one edge of $C_v$ belongs to $E'$. Since every edge of $H$
  may belong to at most one graph in $\set{C_u\mid u\in V(X)}$,
  $|E'|\geq |\tset_A'|\geq \kappa'/2$ must hold. Similarly, if
  $|\tset_B'|\geq \kappa'/2$, $|E'|\geq \kappa'/2$.

  Therefore, we assume from now on that
  $|\tset_A'|,|\tset_B'|<\kappa'/2$, and so
  $|\tset_A''|,|\tset_B''|\geq \kappa'/2$. Let $U_A=\set{v\in V(X)\mid
    t_v\in \tset_A''}$, and $U_B=\set{v\in V(X)\mid t_v\in
    \tset_B''}$. Since $X$ is an $\alpha$-expander, there is a set
  $\pset$ of at least $\kappa'/2$ paths connecting the vertices of
  $U_A$ to the vertices of $U_B$ in $X$, such that the edge-congestion
  of $\pset$ is at most $1/\alpha$. Since the maximum vertex degree in
  $X$ is $\Delta'$, by sending $\alpha/\Delta'$ flow units along each
  path in $\pset$, we obtain a flow from the vertices in $U_A$ to the
  vertices of $U_B$ of value at least
  $\frac{\kappa'\alpha}{2\Delta'}$, where the flow across each vertex
  is at most $1$. From the integrality of flow, there is a set
  $\pset'$ of at least $\frac{\kappa'\alpha}{2\Delta'}$ node-disjoint paths in
  $X$, where each path connects a vertex of $U_A$ to a vertex of
  $U_B$. We now build a new set $\pset^*$ of
  $\frac{\kappa'\alpha}{2\Delta'}$ paths in graph $H$, connecting the
  vertices of $\tset_A''$ to the vertices of $\tset_B''$, as follows.

  Consider some path $P\in \pset'$, and assume that its endpoints are
  $v$ and $u$, so $t_v\in \tset_A''$ and $t_u\in \tset_B''$. We build
  a new graph $H_P$, that includes, for every edge $e\in P$, the path
  $P_e$ into which $e$ is embedded, and for every vertex $v'\in P$,
  the sub-graph $C_{v'}$, where $v'$ is embedded into $C_{v'}$. It is
  easy to see that $H_P$ contains a path connecting $t_v$ to
  $t_u$. Let $P^*$ be any such path. We then set $\pset^*=\set{P^*\mid
    P\in \pset'}$. Since the paths in $\pset'$ are node-disjoint, and
  the embedding of $X$ into $H$ has congestion $\eta$, every edge of
  $H$ belongs to at most $\eta$ paths in $\pset^*$. Since every path
  in $\pset^*$ must contain an edge in $E'$, $|E'|\geq
  \frac{\kappa'\alpha}{2\eta\Delta'}$. We conclude that $\tset$ is
  $\frac{\alpha}{2\eta\Delta'}$-well-linked, and from
  Theorem~\ref{thm: weak well-linkedness to tw}, the treewidth of $H$
  is $\Omega\left (\frac{\alpha\kappa}{\eta\Delta\Delta'}\right )$.
\end{proof}
\fi

\label{--------------------------------------------sec: building the minor----------------------------------------------}
\section{A Small Treewidth-Preserving  Degree-4 Minor}\label{sec: top minor}


\iffull
In this section we prove the following theorem which gives a
degree-$4$ sparsifier.

\begin{theorem}\label{thm: degree-4 minor}
  There is a randomized algorithm, that, given a graph $G$ of
  treewidth at least $k$, w.h.p. computes a minor $H$ of $G$, such
  that:
\begin{itemize}
\item  the treewidth of $H$ is $\Omega(k/\poly\log k)$;  
\item every vertex has degree at most $4$ in $H$; and
\item $|V(H)|=O(k^4)$.
\end{itemize}
The running time of the algorithm is polynomial in $|V(G)|$ and $k$.
\end{theorem}

In order to prove Theorem~\ref{thm: degree-4 minor}, it is sufficient to
find a subgraph $H$ of $G$, with $\tau(H)=O(k^4)$, such that the
maximum vertex degree in $H$ is at most $4$, and the treewidth of $H$
is $\Omega(k/\poly\log k)$. Indeed, by replacing every maximal
$2$-path in $H$ with an edge connecting its endpoints, we obtain the
desired minor.
\fi

We start by applying Theorem~\ref{thm: strong PoS system} with
$r=\gKRV(k)$ and $h=\Omega(k/\poly\log k)$, so that $h$ is an even
integer, and $\frac{k}{\log^{c'}k}>chr^{48}$ holds, where $c,c'$ are
the constants from Theorem~\ref{thm: strong PoS system}. Let
$(\sset,\bigcup_{i=1}^{r-1}\pset_i)$ be the resulting path-of-sets
system.

Our next step is to construct an expander graph $X$ on $h$ vertices,
and to embed it into a sub-graph $H$ of $G$.  Following the previous
work on routing
problems~\cite{RaoZhou,Andrews,Chuzhoy11,ChuzhoyL12,ChekuriE13}, we
will embed $X$ into $G$ using the cut-matching game, and the
path-of-sets system $(\sset,\bigcup_{i=1}^{r-1}\pset_i)$. We start
with an intuitive high-level description of the algorithm.  For each
$1\leq i\leq r$, let $\qset_i$ be any set of node-disjoint paths
connecting $A_i$ to $B_i$ in $G[S_i]$ (this set exists due to the
linkedness of $(A_i,B_i)$ in $G[S_i]$). Let $\hset$ be the set of $h$
paths, obtained by concatenating
$\qset_1,\pset_1,\qset_2,\pset_2,\ldots,\pset_{r-1},\qset_r$. We
denote $\hset=\set{P_1,\ldots,P_{h}}$. The high-level idea is to
construct an expander $X$ over a set $V=\set{v_1,\ldots,v_{h}}$ of $h$
vertices, and to embed it into $G$ using the cut-matching game, as
follows. For each $1\leq i\leq h$, we embed $v_i$ into $P_i$, that is,
$C_{v_i}=P_i$. We construct the edges of the expander, and embed them
into $G$, using the cut-matching game, where for each $1\leq j\leq r$,
we use cluster $S_j\in \sset$ to route the $j$th matching, as
follows. A partition $(Y,Z)$ of the vertices of $X$ computed by the
cut player naturally defines a partition $(\hset_Y,\hset_Z)$ of the
paths in $\hset$ into two equal-sized subsets, which in turn defines a
partition $(A_j',A_j'')$ of $A_j$ into two equal-sized subset. Using
the fact that $A_j$ is node-well-linked inside $G[S_j]$, we can find a
set $\bset_j$ of node-disjoint paths in $G[S_j]$ connecting $A_j'$ to
$A_j''$. This set of paths defines a matching $M_j$ between the paths
in $\hset_Y$ and $\hset_Z$, and hence between the vertices of $Y$ and
$Z$ in $X$. We view this matching as the response of the matching
player. After $\gKRV(h)\leq \gKRV(k)=r$ iterations, we obtain an
expander $X$ and its embedding with congestion $2$ into
$G$. Intuitively, we would like to define $H$ as the set of all edges
and vertices of $G$ used in this embedding, that is, the union of the
paths in $\hset$ and $\bigcup_{j=1}^r\bset_j$. It is easy to see that
the maximum vertex degree in $H$ is at most $4$, since in each cluster
$S_j$ we only route 2 sets of node-disjoint paths: $\qset_j$ and
$\bset_j$. However, $\tau(H)$ may not be bounded by $O(k^4)$. In
particular, the paths in $\qset_j$ and $\bset_j$ may intersect at many
vertices. In order to overcome this difficulty, we can use
Theorem~\ref{thm: 2-flow-main} to find new sets $\qset_j'$ and
$\bset_j'$ of paths, routing the same pairs of vertex subsets, such
that, if $J=J(\qset_j'\cup \bset_j')$, then $\tau(J)=O(h^4)$. However,
this re-routing changes the paths in $\hset$, and therefore the
mapping between the vertices in sets $A_{j'}$ for $j'\neq j$ and the
vertices in $X$ may be changed. Therefore, we need to execute this
procedure more carefully. In particular, we apply Theorem~\ref{thm:
  2-flow-main} in the graph $G[S_j]$ after each iteration $j$ of the
cut-matching game; for iteration $j+1$ we exploit the
node-well-linkedness of the set $A_{j+1}$ in $G[S_{j+1}]$ to maintain
consistency in the mapping of paths in $\hset$ to the vertices of the
expander.  \iffull We now provide a
formal description of the embedding procedure.

We will gradually construct the set $\hset$ of paths over the course
of $\gKRV(h)$ iterations. For each $1\leq j\leq \gKRV(h)$, at the
beginning of the $j$th iteration, we are given a set $\hset^j$ of $h$
disjoint paths, connecting the vertices of $A_1$ to the vertices of
$A_j$, and a bijection $f: \hset^j\rightarrow V(X)$. At the beginning,
$\hset^1$ consists of $h$ paths, where each path consists of a single
distinct vertex of $A_1$, and the mapping $f:\hset^1\rightarrow V(X)$
is an arbitrary bijection. We also start with a graph $X$ on $h$
vertices, and $E(X)=\emptyset$. For $1\leq j\leq \gKRV(h)$, the $j$th
iteration is executed as follows.

We use the cut player on the current graph $X$ to compute a partition
$(Y_j,Z_j)$ of $V(X)$ into two equal-sized subsets. This naturally
defines a partition $(\hset^j_Y,\hset^j_Z)$ of $\hset^j$ where
$\hset^j_Y$ contains all paths $P\in \hset^j$, such that $f(P)\in
Y_j$. In turn, this gives a partition $(A_j',A_j'')$ of $A_j$, where a
vertex $v\in A_j$ belongs to $A_j'$ iff the path $P\in\hset^j$ on which $v$ lies
belongs to $\hset^j_Y$. Since the set $A_j$ of vertices is
node-well-linked in $G[S_j]$, there is a collection of node-disjoint
paths routing $(A_j',A_j'')$ in $G[S_j]$. Since $A_j$ and $B_j$ are
linked in $G[S_j]$, there is a collection of node-disjoint paths
routing $(A_j,B_j)$ in $G[S_j]$. From Theorem~\ref{thm: 2-flow-main},
we can find a set $\bset_j$ of paths routing $(A_j',A_j'')$, and a
set $\qset_j$ of paths routing $(A_j,B_j)$ in $G[S_j]$, such that, if
$J=J(\bset_j\cup \qset_j)$, then the maximum vertex degree in $J$ is
bounded by $4$, the degree of every vertex in $A_j\cup B_j$ is at most
$3$, and $\tau(J)\leq O(h^4)$. We let $\hset^{j+1}$ be the
concatenation of the paths in $\hset^j$, $\qset_j$, and $\pset_j$. In
order to define the mapping $f:\hset^{j+1}\rightarrow V(X)$, for each
$P\in \hset^{j+1}$, let $P'\in \hset^j$ be the sub-path of $P$. Then
we set $f(P)=f(P')$. Notice that the set $\bset_j$ of paths defines a
complete matching between the vertices in $A_j'$ and $A_j''$, and by
extension, a complete matching between the paths in $\hset^j_Y$ and
$\hset^j_Z$, which in turn naturally defines a matching $M_j$ between
$Y_j$ and $Z_j$ in $X$. We add the edges of the matching $M_j$ to
$X$. Each edge $e=(v_i,v_{i'})\in M_j$ is mapped to the corresponding
path in $P_e\in \bset_j$, that connects the unique vertex of $P_i\cap
A_j$ to the unique vertex of $P_{i'}\cap A_j$.

Let graph $H$ be the union of all paths in $\hset^{\gKRV(h)}$ and
$\bigcup_{j=1}^{\gKRV(h)}\bset_j$. Then it is easy to see that the
maximum vertex degree in $H$ is at most $4$, and $\tau(H)\leq
O(h^4\gKRV(h))=O(k^4)$. Moreover, every edge of $H$ belongs to at most
one path in $\hset$, and at most one path in
$\bigcup_{j=1}^{\gKRV(h)}\bset_j$. Therefore, we have constructed an
embedding of an $h$-vertex $\alphaCMG(h)$-expander $X$, whose maximum
vertex degree is $\gKRV(h)$, into $H$ with congestion 2, and so from
Claim~\ref{claim: embed an expander treewidth}, the treewidth of $H$
is at least $\Omega(h/\gKRV(h))=\Omega(k/\poly\log k)$.
\fi

\label{---------------------------------------sec: degree-3----------------------------------------------}
\section{Building a Degree-$3$ Minor}\label{sec: degree-3}
In this section we complete the proof of Theorem~\ref{thm:
  main-topological-minor}.  We start with an informal overview to help
understand the high-level plan.  A reader may wish to skip it and go
directly to the formal proof.

\subsection{Overview}
We use an algorithm, similar to the one used in Section~\ref{sec: top
  minor}, in order to embed an expander into $G$, using the
path-of-sets system. The main difference is that, instead of embedding
a single expander $X$, we will embed $N$ expanders $X_1,\ldots,X_N$
where $N = \Theta(\log k)$. For this purpose we start with a longer
path-of-sets system $(\sset,\bigcup_{i=1}^{r-1}\pset_i)$ with
parameters $h=k/\poly\log k$ and $r=O(\log^3k)$ and partition it into
$N=O(\log k)$ smaller path-of-sets systems with parameters
$r^*=\gammaKRV(h)$ and $h$ (hence $r \simeq N r^*$). For $1\leq i\leq N$, we
embed an expander $X_i$ into the $i$'th path-of-sets system using the
approach in the preceding section. Recall that for each cluster $S_i$,
we construct two sets of paths contained in $G[S_i]$: one set,
$\rset_i$, that we call red paths, routes $(A_i,B_i)$, and another
set, $\bset_i$, that we call blue paths, routes $(A_i',A_i'')$, where
$(A_i',A_i'')$ is the partition of $A_i$ defined by the cut
player. Let $H_i$ be the topological minor of $G[S_i]$ obtained by
taking the union of the paths in $\rset_i$ and $\bset_i$, and
suppressing all degree-$2$ vertices, except for $A_i\cup B_i$. We
assume that $H_i$ is minimal in the following sense: for each edge $e$
of $H_i$, either $(A_i,B_i)$ or $(A_i',A_i'')$ is not routable in
$H_i\setminus \set{e}$. Abusing the notation, we assume that $\rset_i$
and $\bset_i$ are the sets of the red and the blue paths, routing
$(A_i,B_i)$ and $(A_i',A_i'')$, respectively in $H_i$. Notice that
every vertex of $H_i$ must lie on some red path. Let $\hset$ be the
set of $h$ paths obtained by concatenating the paths in
$\rset_1,\pset_1,\ldots,\pset_{Nr^*-1},\rset_{Nr^*}$, and let $H$ be
the topological minor of $G$ obtained by taking the union of the
graphs $H_i$ and the paths $\bigcup_{i=1}^{Nr^*-1}\pset_i$.
We say that an edge of $H$ is
a red edge if it belongs to a red path and no blue paths; it is a blue
edge if it belongs to a blue path and no red paths; and it is a
red-blue edge if it belongs to both a red and a blue path. We can view
the $N$ different expanders as sharing the same vertex set, where the
vertices correspond to the paths in $\hset$. 
Consider a vertex $v$ of
degree $4$ in graph $H$; it must be incident to two red edges and two
blue edges.  In order to reduce the degree to $3$, we use random
sampling to pick one of the two blue edges incident to $v$ and eliminate
it. After this step the degree of every vertex is at most $3$.  Let
$H^*$ be this final topological minor of $G$. The heart of the
analysis is to show that $H^*$ has treewidth $\Omega(k/\poly\log k)$.
This is done by showing that the set $A=A_1$ of vertices remains
$\alpha$-well-linked in $H^*$, for $\alpha=\Omega(1/\poly\log k)$, and
applying Theorem~\ref{thm: weak well-linkedness to tw}.
 
We start by observing that the set $A$ of vertices is
$\alphaWL$-well-linked in $H$, for some constant $\alphaWL$. This is
shown by using the embeddings of the expanders $X_1,\ldots,X_N$ into
$H$.  Next, we carefully partition each path in $\hset$ into a
collection of disjoint segments. Intuitively, each segment of a path
$P\in \hset$ is a sub-path of $P$ of length $\Theta(\poly\log h)$. We
then contract each such segment $\sigma$ into a super-node
$v_{\sigma}$. Let $F$ be this contracted graph, and let $F^*$ be the
corresponding contracted graph of $H^*$. Equivalently, $F^*$ is
obtained from $F$ by deleting all the edges in $E(H)\setminus E(H^*)$.

Each vertex of $A$ belongs to a distinct contracted segment, and is
associated with the corresponding super-node in $F$. We do not
distinguish between the vertices of $A$ and their corresponding
super-nodes. It is easy to see that $A$ remains $\alphaWL$-well linked
in $F$ since we only contracted edges. The most crucial property of
the contracted graph $F$ is that the value of the minimum cut in $F$
is at least $\Omega(\log |V(F)|)$. This allows us to use arguments
similar to those used in Karger's sampling
technique~\cite{Karger-skeleton} to show that all cuts are
approximately preserved in $F^*$. In particular, the vertices of $A$
remain $\alphaWL/32$-well-linked in $F^*$. Since the length of every
segment used in the construction of the contracted graph $F$ is
$O(\poly\log h)$, this implies that the vertices of $A$ are
$\alpha$-well-linked in $H^*$, for $\alpha=\Omega(1/\poly\log k)$.
The most challenging part of the proof is to set up the partition of
the paths in $\hset$ into segments, so that in the resulting
contracted graph $F$, the value of the minimum cut is $\Omega(\log
|V(F)|)$.  At a high-level, the proof proceeds as follows. Assume for
contradiction, that there is a partition $(X,Y)$ of $V(F)$ with
$X,Y\neq \emptyset$, and $|E_F(X,Y)|<N$. Let $X'\subseteq V(H)$ be
obtained from $X$ by un-contracting all super-nodes in $X$, and let
$Y'\subseteq V(H)$ be obtained from $Y$ similarly. Then $(X',Y')$ is a
partition of $V(H)$, and $|E_{H}(X',Y')|<N$. Assume first that there
are two paths $P,P'\in \hset$, such that $P\subseteq H[X']$ and
$P'\subseteq H[Y']$. We then use the embeddings of the expanders
$X_1,\ldots,X_N$ to argue that $|E_H(X',Y')|\geq N$, reaching a
contradiction. Therefore, we can assume w.l.o.g. that no path of
$\hset$ is contained in $H[X']$. We next show that for some $1\leq
i^*\leq Nr^*$, partition $(X',Y')$ of $V(H)$ defines a partition
$(X^*,Y^*)$ of $V(H_{i^*})$, such that $|X^*|,|Y^*|>200N^4$, while
$|E_{H_{i^*}}(X^*,Y^*)|<N$. We then consider the segments of the red
paths in $\rset_{i^*}$ and the blue paths in $\bset_{i^*}$ that are
contained in $H_{i^*}[X^*]$. Let $\rset^*$ denote the corresponding
segments of the red paths, and $\bset^*$ the corresponding segments of
the blue paths. Using Theorem~\ref{thm: 2-flow-minor}, we show that
there is some edge $e\in H[X^*]$, such that we can still route the
endpoints of the paths in $\rset^*$ to each other, and the endpoints
of the paths in $\bset^*$ to each other, even after deleting $e$ from
$H[X^*]$. This new routing implies that we can route both
$(A_{i^*},B_{i^*})$ and $(A_{i^*}',A_{i^*}'')$ in
$H_{i^*}\setminus\set{e}$, contradicting the minimality of $H_{i^*}$.

\subsection{Proof of Theorem~\ref{thm: main-topological-minor}}
We set $r=2^{15}\log k\cdot \gammaKRV(k) =\Theta(\log^3 k)$, and
$h=\Omega(k/\poly\log k)$, so that $h$ is an even integer, and
$k/\log^{c'}k>chr^{48}$, where $c$ and $c'$ are the constants from
Theorem~\ref{thm: strong PoS system}. We assume w.l.o.g. that $k$ is
large enough, so $h>72\log k$ and $h>\gammaKRV(h)$. We then apply
Theorem~\ref{thm: strong PoS system} to graph $G$, with parameters $r$
and $h$, to obtain a strong path-of-sets system
$(\sset,\bigcup_{i=1}^{r-1}\pset_i)$ of height $h$ and width $r$.

Let $r^*=\gammaKRV(h)$, and let $N=\ceil{3072\log(10h^4\cdot r^*)}$;
it is easy to see that $N=\Theta(\log h)$. We will assume
w.l.o.g. that $h$ is large enough, so $N>1536\log(10h^4\cdot r^*\cdot
N)$ holds.  Finally, we let $r'=N \cdot r^*$. Note that $r'=r^* \cdot
\ceil{3072\log(10h^4\cdot r^*)}\leq 2^{15}\gammaKRV(h) \log h<r$.

We construct a new, smaller, path-of-sets system, of height $h$ and
width $r'$, using the clusters $\sset'=(S_1,\ldots,S_{r'})$, and the
sets $\pset_i$ of paths, for $1\leq i\leq r'-1$; in other words
we restrict attention to the first $r'$ clusters from the initial path-of-sets
system. Abusing notation, we denote $r'$ by $r$ and $\sset'$ by $\sset$.


We denote by $G'$ the following minor of $G$: start with the union of
$G[S_i]$ for all $1\leq i\leq r$; for each path $P\in
\bigcup_{i=1}^{r-1}\pset_i$, add an edge connecting the endpoints of
$P$ to $G'$. We denote by $E_i$ the set of edges corresponding to the
paths in $\pset_i$. Equivalently, we obtain $G'$ from graph $\left (\bigcup_{S_i\in \sset}G[S_i]\right )\cup\left (\bigcup_{j=1}^{r-1}\pset_j\right )$ by suppressing degree-$2$ internal nodes on the paths in $\bigcup_{i=1}^{r-1}\pset_i$.  It
is now enough to find a topological minor $H^*$ of $G'$ whose treewidth is
$\Omega(k/\poly\log k)$, maximum vertex degree is $3$, and
$|V(H^*)|=O(k^4)$.  We do so via the following theorem:

\begin{theorem}\label{theorem: finding a minor in a super-cluster}
  There is an efficient randomized algorithm, that 
  finds a topological minor $H^*$ of $G'$, such that
  w.h.p.: 
  
  \begin{itemize}
  \item $|V(H^*)|=O(h^4\cdot r)$; 
  \item  The maximum
  vertex degree in $H^*$ is $3$; 
  
  \item $A_1\subseteq V(H^*)$; and
  \item The set $A_1$ of vertices is
  $\alpha$-well-linked in $H^*$, for $\alpha=\Omega(1/\log^7 k)$.
\end{itemize}
\end{theorem}

Theorem~\ref{thm: main-topological-minor} follows easily from
Theorem~\ref{theorem: finding a minor in a super-cluster}.  The
desired topological minor of $G$ is $H^*$. The only property that is
left to verify is that $\tw(H) = \Omega(k/\polylog k)$ which follows
from $\alpha$-well-linkedness of $A_1$ in $H^*$.  Indeed,
Theorem~\ref{thm: weak well-linkedness to tw} implies that $\tw(H) =
\Omega(\alpha |A_1|/3) = \Omega(k/\polylog k)$ since
$|A_1| = h = \Omega(k/\polylog k)$, $\alpha = \Omega(1/\log^7 k)$
and $H^*$ has maximum degree $3$.
From now on we focus on proving Theorem~\ref{theorem: finding a minor in a super-cluster}.

In order to simplify the notation, we refer to the graph $G'$ as $G$.
Recall that we are given a path-of-sets system
$(\sset=(S_1,\ldots,S_r),\bigcup_{i=1}^{r-1}\pset_i)$ of height $h$
and width $r=Nr^*$ in $G$, where for each $1\leq i< r$, each path in
$\pset_i$ consists of a single edge, and the corresponding set of
edges is denoted by $E_i$. Let $E'=\bigcup_{i=1}^{r-1}E_i$. We denote
$A_1$ by $A$.  Our goal is to construct a topological minor $H^*$ of
$G$, such $|V(H^*)|=O(h^4 r)$, the maximum vertex degree of $H^*$ is
$3$, while ensuring that $A\subseteq V(H^*)$ and it is
$\alpha$-well-linked in $H^*$, w.h.p.

The rest of the proof consists of three steps. In the first step, we
define the sets $\bset_i,\rset_i$ of paths for $1\leq i\leq r$ by
playing the cut-matching games; in the second step we partition the
resulting red paths into segments; and in the third step we complete
the proof of the theorem.

\paragraph{Step 1: Cut-Matching Games}
In this step we construct $N$ expanders $X_1,\ldots,X_N$, and embed each of them separately into $G$. For each $1\leq i\leq N$, let $\sset_i=(S_{(i-1)r^*+1},\ldots,S_{ir^*})$, let $E^i=\bigcup_{j=(i-1)r^*+1}^{ir^*-1}E_j$, and let $\hat E^i=E_{ir^*}$ (for $i=N$, $\hat E^i=\emptyset$). Let $G_i$ be the graph obtained from the union of $G[S_j]$ for all $S_j\in \sset_i$ and the edges in $E^i$. For each $1\leq i\leq N$, we embed the expander $X_i$ into $G_i$, using the cut-matching game, as follows. For convenience, we denote $(i-1)r^*$ by $z$.

We will gradually construct a set $\hset_i$ of paths over the course of $r^*$ iterations. For each $1\leq j\leq r^*$, at the beginning of the $j$th iteration, we are given a set $\hset^j$ of $h$ disjoint paths, connecting the vertices of $A_{z+1}$ to the vertices of $A_{z+j}$, and a bijection $f: \hset^j\rightarrow V(X_i)$. At the beginning, $\hset^1$ consists of $h$ paths, each of which consists of a single distinct vertex of $A_{z+1}$, and the mapping $f:\hset^1\rightarrow V(X_i)$ is an arbitrary bijection. We also start with a graph $X_i$ on $h$ vertices, and $E(X_i)=\emptyset$. For $1\leq j\leq r^*$, the $j$th iteration is executed as follows. 

We use the cut player on the current graph $X_i$ to find a partition $(Y_j,Z_j)$ of $V(X_i)$ into two equal-sized subsets. This naturally defines a partition $(\hset^j_Y,\hset^j_Z)$ of $\hset^j$ where $\hset^j_Y$ contains all paths $P\in \hset^j$, such that $f(P)\in Y_j$. In turn, this gives a partition $(A_{z+j}',A_{z+j}'')$ of $A_{z+j}$, where a vertex $v\in A_{z+j}$ belongs to $A_{z+j}'$ iff the path $P$ on which $v$ lies belongs to $\hset^j_Y$. Since the set $A_{z+j}$ of vertices is node-well-linked in $G[S_{z+j}]$, there is a collection of node-disjoint paths routing $(A_{z+j}',A_{z+j}'')$ in $G[S_{z+j}]$. Since $A_{z+j}$ and $B_{z+j}$ are linked in $G[S_{z+j}]$, there is a collection of node-disjoint paths routing $(A_{z+j},B_{z+j})$ in $G[S_{z+j}]$. From Theorem~\ref{thm: 2-flow-main}, we can find a set $\bset_{z+j}'$ of paths routing  $(A_{z+j}',A_{z+j}'')$ , and a set $\rset_{z+j}'$ of paths routing $(A_{z+j},B_{z+j})$ in $G[S_{z+j}]$, such that, if $J=J(\bset_{z+j}'\cup\rset_{z+j}')$, then the maximum vertex degree in $J$ is bounded by $4$, the degree of every vertex in $A_{z+j}\cup B_{z+j}$ is at most $3$, and $\tau(J)\leq 8h^4+8h$. 
We will assume that $J$ is a minimal graph in which $(A_{z+j}',A_{z+j}'')$ and $(A_{z+j},B_{z+j})$ are both routable: that is, for every edge $e\in E(J)$, either $(A_{z+j}',A_{z+j}'')$, or $(A_{z+j},B_{z+j})$ are not routable in $J\setminus \set{e}$. We let $H_{z+j}$ be the graph obtained from $J$ by replacing every maximal $2$-path that does not contain the vertices of $A_{z+j}\cup B_{z+j}$ as inner vertices, by an edge connecting its two endpoints. Then $|V(H_{z+j})|\leq 8h^4+8h\leq 10h^4$, every vertex of $H_{z+j}$ has degree at most $4$, while the vertices in $A_{z+j}\cup B_{z+j}$ have degree at most $3$; there is a set $\bset_{z+j}$ of paths routing  $(A_{z+j}',A_{z+j}'')$ , and a set $\rset_{z+j}$ of paths routing $(A_{z+j},B_{z+j})$ in $H_{z+j}$, and for every edge $e\in E(H_{z+j})$, either $(A_{z+j}',A_{z+j}'')$, or $(A_{z+j},B_{z+j})$ are not routable in $H_{z+j}\setminus \set{e}$.
We call the paths in $\rset_{z+j}$ \emph{red paths}, and the paths in $\bset_{z+j}$ \emph{blue paths}. An edge that belongs to a red path, but no blue paths is called a red edge. An edge the belongs to a blue path but no red paths is called a blue edge. An edge that lies on a red and a blue path is called a red-blue edge. Notice that a vertex of $H_{z+j}$ has degree $4$ only if it is incident on two blue edges. Each vertex in $A_{z+j}$ serves as a source of a red path and a source or a destination of a blue path, so it can only be incident on at most two edges in $H_{z+j}$. A vertex $v\in B_{z+j}$ serves as a destination of a red path; its degree is at most $3$, and it is equal to $3$ only if $v$ is incident on two blue edges.

We let $\hset^{j+1}$ be the concatenation of the paths in $\hset^j$, $\rset_{z+j}$, and $E_{z+j}$. In order to define the mapping $f:\hset^{j+1}\rightarrow V(X_i)$, for each $P\in \hset^{j+1}$, let $P'\in \hset^j$ be the sub-path of $P$. Then we set $f(P)=f(P')$. Notice that the set $\bset_{z+j}$ of paths defines a matching between the paths in $\hset^j_Y$ and $\hset^j_Z$, which in turn naturally defines a matching $M_j$ between $Y_j$ and $Z_j$ in $X_i$. We add the edges of the matching $M_j$ to $X$. Each edge $e=(v_i,v_{i'})\in M_j$ is mapped to the corresponding path in $\bset_{z+j}$, that connects the unique vertex in $A_j\cap f^{-1}(v_i)$ to the unique vertex in $A_j\cap f^{-1}(v_{i'})$. 

Finally, we set $\hset_i=\hset^{r^*}$. Let $\tilde H_i$ be the union of the graphs $H_{z+1},\ldots,H_{z+r^*}$, and the edges $E^i$. Then we have defined an $\alphaCMG(h)$-expander $X_i$ on $h$ vertices with maximum vertex degree $\gammaKRV(h)$, and embedded it with congestion $2$ into $\tilde H_i$, where each vertex of $X_i$ is embedded into a distinct path in $\hset_i$.

Let $H$ be the union of the graphs $\tilde H_i$, for $1\leq i\leq N$
and $\bigcup_{i=1}^{N-1}\hat E^i$, and let $\hset$ be the
concatenation of $\hset_1,\hat E_1,\ldots,\hat E_{N-1},\hset_N$. We
will sometimes refer to the paths in $\hset$ as red paths. All
vertices in $H$ have degree at most $4$, and, as observed before, a
vertex of $H$ may have degree $4$ only if it is incident on exactly
two blue edges. Every vertex in $A$ has degree at most $2$.  Our final
graph $H^*$ is obtained from $H$ as follows: for each vertex $v\in
V(H)$ that is incident on two blue edges, we independently choose one
of these two blue edges at random. Each blue edge that has been chosen
by at least one vertex is then deleted from the graph. This final
graph is denoted by $H^*$. Notice that each edge $e=(u,v)$ may be
deleted from $H$ due to the choice made by $u$, or the choice made by
$v$; the overall probability that $e$ is not deleted is at least
$1/4$. Moreover, if $e$ and $e'$ do not share endpoints, then the
events that $e$ is deleted and that $e'$ is deleted are independent.

 It is immediate to see that  $|V(H^*)|\leq Nr^*\cdot O(h^4)=O(r h^4)$; the vertices of $A$ are contained in $V(H^*)$, and the maximum vertex degree in $H^*$ is $3$. It now only remains to prove that w.h.p. the vertices of $A$ are $\alpha$-well-linked in $H^*$, for some $\alpha=\Omega(1/\log^7 k)$. We do so in the next two steps, using the following claim.

\begin{claim}\label{claim: A well-linked}
The set $A$ of vertices is $\alphaWL$-well-linked in $H$, where $\alphaWL=\min\set{\half,\frac{N\cdot\alphaKRV(h)}{4\gammaKRV(h)}}=\Omega(1)$.
\end{claim}

\begin{proof}
Let $(Y,Z)$ be any partition of $V(H)$, $A_Y=A\cap Y$, $A_Z=A\cap Z$, and assume that $|A_Y|\leq |A_Z|$. We denote $|A_Y|$ by $\kappa$. The it is enough to show that $|E_H(Y,Z)|\geq \alphaWL \kappa$. We partition the set $A_Y$ of vertices into subsets: $A''_Y$ contains all vertices $v\in A$, such that the unique path $P\in\hset$ on which $v$ lies is contained in $Y$, and $A'_Y$ contains the remaining vertices. We partition $A_Z$ into $A'_Z$ and $A''_Z$ similarly. Assume first that $|A'_Y|\geq \kappa/2$. Then $|E_H(Y,Z)|\geq \kappa/2$, since for every vertex $v\in A'_Y$, the corresponding path $P\in \hset$ contributes at least one edge to $E_H(Y,Z)$. Similarly, if $|A'_Z|\geq \kappa/2$, then $|E_H(Y,Z)|\geq \kappa/2$. From now on we assume that $|A'_Y|,|A'_Z|<\kappa/2$, and so $|A_Z''|\geq |A_Y''|\geq \kappa/2$.

Let $\yset\subseteq \hset$ be the set of all the paths $P$, such that the first vertex of $P$ belongs to $A''_Y$. Define $\zset\subseteq \hset$ similarly for $A''_Z$. 

Fix some $1\leq i\leq N$, and consider the expander $X_i$.
We define two subsets of vertices of $X_i$: $Y^*$ contains all vertices $v$ that are embedded into the sub-paths of $\yset$, and $Z^*$ contains all vertices that are embedded into the sub-paths of $\zset$. Since $X_i$ is an $\alphaKRV(h)$-expander, there are $\alphaKRV(h)\cdot |Y^*|\geq  \alphaKRV(h)\cdot \kappa/2$ edge-disjoint paths connecting the vertices of $Y^*$ to the vertices of $Z^*$ in $X_i$. Since the maximum vertex degree in $X_i$ is $\gammaKRV(h)$, there is a collection $\lset_i$ of at least $\frac{\alphaKRV(h)}{\gammaKRV(h)}\cdot \frac{\kappa}{2}$ node-disjoint paths in $X_i$ connecting the vertices of $Y^*$ to the vertices of $Z^*$. We construct a collection $\lset_i'$ of paths, connecting the vertices of $V(\yset)$ to the vertices of $V(\zset)$, such that $\lset_i'\subseteq \tilde H_i$, and each edge of $\tilde H_i$ participates in at most two such paths. For each path $P\in \lset_i$, we build a graph $G_P$ as follows: for each edge $e\in E(P)$, the graph includes the blue path of $\tilde H_i$ into which the edge $e$ is embedded, and, for each vertex $v\in E(P)$, the graph includes the red path $P_v\in \hset_i$ into which $v$ is embedded. It is then easy to see that $G_P$ contains a path $P'\subseteq \tilde H_i$ connecting a vertex on some path $Q\in \yset$ to a vertex on some path $Q'\in \zset$. We let $\lset_i'=\set{P'\mid P\in \lset}$. Since each edge of $\tilde H_i$ belongs to at most one red path and at most one blue path, and the paths of $\lset_i$ are node-disjoint, each edge of $\tilde H_i$ is contained in at most two paths of $\lset_i'$. Let $\lset=\bigcup_{i=1}^N\lset_i'$.  Then $\lset$ contains $\frac{N\cdot\alphaKRV(h)}{\gammaKRV(h)}\cdot\frac{\kappa}{2}$ paths, where each path connects a vertex in $V(\yset)$ to a vertex in $V(\zset)$, and each edge of $H$ belongs to at most two such paths. Each path of $\lset$ connects a vertex of $X$ to a vertex of $Y$, and so $|E_H(X,Y)|\geq 
\frac{N\cdot\alphaKRV(h)}{4\gammaKRV(h)}\cdot \kappa$.
\end{proof}

\paragraph{Step 2: Partitioning the Red Paths}
In this step, we will define a collection $\Sigma_P$ of disjoint segments for every path $P\in\hset$.

Consider any such path $P\in \hset$. A sub-path $P'$ of $P$ is called a \emph{heavy sub-path} iff for some $1\leq i\leq Nr^*$, $P'$ contains at least $200N^4=\Theta(\log^4h)$ vertices that belong to $H_i$. 

If $P$ contains no heavy sub-paths, then $\Sigma_P=\set{P}$. Notice that $P$ contains at most $Nr^*\cdot O(\log^4h)=O(\log^7h)$ vertices in this case.
 Otherwise, we perform a number of iterations. In each iteration, we start with some heavy sub-path $P'$ of $P$, where at the beginning of the first iteration, $P'=P$. 
Let $P''$ be the minimum-length heavy sub-path of $P'$ containing the first vertex of $P'$. If $P'\setminus P''$ is a heavy path, then we add $P''$ to $\Sigma_P$, delete all vertices of $P''$ from $P'$, and continue to the next iteration. Otherwise, we add $P'$ to $\Sigma_P$ and finish the algorithm. Notice that in any case, the length of every path added to $\Sigma_P$ is at most $Nr^*\cdot O(\log^4h)=O(\log^7h)$. Overall, for each path $P\in \hset$, we obtain a partition of $P$ into disjoint sub-paths of length at most $O(\log^7h)$ each. Moreover, if $|\Sigma_P|>1$, then each path in $\Sigma_P$ is a heavy sub-path of $P$. Let $\Sigma=\bigcup_{P\in \hset}\Sigma_P$.

We obtain a contracted graph $F$ from $H$ by contracting, for each $\sigma\in \Sigma$, the vertices of $\sigma$ into a single super-node $v_{\sigma}$. For every vertex $u\in A$, let $g(u)$ be the super-node $v_{\sigma}$ such that $u\in V(\sigma)$. Notice that for $u\neq u'$, $g(u)\neq g(u')$. Let $U=\set{g(u)\mid u\in A}$.
Since,  from Claim~\ref{claim: A well-linked}, the vertices of $A$ are $\alphaWL$-well-linked in $H$, the vertices of $U$ are $\alphaWL$-well-linked in $F$. Since every vertex of $H$ must belong to some red path, $V(F)=\set{v_{\sigma}\mid \sigma\in \Sigma}$.


We define a graph $F^*$ from $H^*$, by similarly contracting all segments in $\bigcup_{P\in\hset}\Sigma_P$ into super-nodes. Equivalently, graph $F^*$ is obtained from $F$, by deleting all edges in $E(H)\setminus E(H^*)$. 
We prove the following claim.

\begin{claim}\label{claim: well-linkedness in sampled graph}
Set $U$ is $\alphaWL/32$-well-linked in $F^*$ w.h.p.
\end{claim}

Assume first that the above claim is correct. We claim that $A$ is $\alpha$-well-linked in $H^*$, for $\alpha=\Omega(1/\log^7h)$. Indeed, let $(X,Y)$ be any partition of vertices of $H^*$. Let $A_X=A\cap X$, $A_Y=A\cap Y$, and $E^*=E_{H^*}(X,Y)$. Assume w.l.o.g. that $|A_X|\leq |A_Y|$. It is enough to prove that $|E^*|\geq \alpha |A_X|$. In order to prove this, we show that there is a set $\qset'$ of $|A_X|$ paths in $H^*$ connecting the vertices of $A_X$ to the vertices of $A_Y$ with edge-congestion at most $1/\alpha$.

Let $U_X=\set{g(v)\mid v\in A_X}$ and $U_Y=\set{g(v)\mid v\in Y_X}$. Since set $U$ is $\alphaWL/32$-well-linked in $F^*$, there is a set $\qset$ of $|U_X|=|A_X|$ paths in $F^*$, such that each path connects a distinct vertex of $U_X$ to a distinct vertex of $U_Y$, and each edge of $F^*$ participates in at most $32/\alphaWL$ such paths. We use the paths in $\qset$ in a natural way, in order to define the set $\qset'$ of paths in graph $H^*$. Let $P\in \qset$ be any such path. Assume that the endpoints of $P$ are $s$ and $t$, and let $s'\in A_X$, $t'\in A_Y'$ be such that $g(s')=s$ and $g(t')=t$. Consider the following sub-graph $H_P$ of $H^*$: start with all the edges that belong to $P$; for each vertex $v_{\sigma}$ on $P$, add the path $\sigma$ to $H_P$. It is easy to see that graph $H_P$ contains a path connecting $s'$ to $t'$. Let $P'$ be any such path. We then set $\qset'=\set{P'\mid P\in \qset}$. Since every vertex $v_{\sigma}$ of $F$ corresponds to a path $\sigma$ of length $O(\log^7h)$ in graph $H$, the degree of each such vertex $v_{\sigma}$ is $O(\log^7h)$. Since the paths in $\qset$ cause edge-congestion at most $32/\alphaWL=O(1)$ in $F^*$, each vertex $v_{\sigma}$ may belong to $O(\log^7h)$ such paths. Therefore, the paths in $\qset'$ cause edge-congestion $O(\log^7h)$ in $H^*$, and $|E^*|\geq \Omega(|A_X|/\log^7h)$. We conclude that $A$ is $\alpha$-well-linked in $H^*$, for $\alpha=\Omega(1/\log^7h)=\Omega(1/\log^7k)$.

\paragraph{Step 3: Finishing the Proof}
In this step we prove Claim~\ref{claim: well-linkedness in sampled graph}. 
We will sometimes refer to a subset $S\subseteq V(F)$ of the vertices of $F$, with $S,V(F)\setminus S\neq \emptyset$, as a \emph{cut}. The value of the cut $S$ is $|\out(S)|$.
The crucial part of the proof is the following claim.

\begin{claim}\label{claim: large degree}
The value of the minimum cut in graph $F$ is at least $N$.
\end{claim}

We prove Claim~\ref{claim: large degree} below, and first complete the proof of Claim~\ref{claim: well-linkedness in sampled graph} using it. 
Let $n'=|V(H)|$. Then $|V(F)|\leq n'\leq 10 h^4\cdot r^*\cdot N$, and since $N>1536\log(10h^4\cdot r^*\cdot N)\geq 1536\log n'$, the value of the minimum cut in $F$ is at least $1536\log n'$. 
The number of edges in $F$ is bounded by $m\leq 4n'\leq 40h^4r^*N=O(h^4\log^3h)$.
We use the following theorem of Karger:

\begin{theorem}[Corollary A.6 in \cite{Karger-skeleton}] \label{thm: number of almost-minimum cuts} Let $G$ be any $n$-vertex graph, and assume that the value of the minimum cut in $G$ is $C$. Then for any half-integral $\beta$, the number of cuts of value at most $\beta C$ in $G$ is bounded by $n^{2\beta}$.
\end{theorem}

Since in graph $F$, the set $U$ of vertices is $\alphaWL$-well-linked, it is enough to show that w.h.p., for any subset $S$ of vertices of $F$, $|\out_{F^*}(S)|\geq |\out_F(S)|/32$. We partition the cuts $S\subseteq V(F)$ into $\ceil{\log m}$ collections $\cset_1,\ldots,\cset_{\ceil{\log m}}$, where for each $1\leq i\leq \ceil{\log m}$, $\cset_i$ contains all cuts $S$ with $2^{i-1}N< |\out_F(S)|\leq 2^iN$; set $\cset_1$ also contains all cuts $S$ with $|\out_F(S)|=N$. Consider now some such collection $\cset_i$. From Theorem~\ref{thm: number of almost-minimum cuts}, $|\cset_i|\leq (n')^{2^{i+1}}$. Consider some set $S\in \cset_i$. Let $S'\subseteq V(H)$ be obtained by un-contracting all super-nodes in $S$, that is, $S'=\bigcup_{v_{\sigma}\in S}V(\sigma)$. Notice that $\out_H(S')=\out_F(S)$, and $\out_{H^*}(S')=\out_{F^*}(S)$. 
Let $E_1(S)\subseteq \out_H(S')$ contain all red and red-blue edges of $\out_H(S')$, and let $E_2(S)=\out_H(S')\setminus E_1(S)$. If $|E_1(S)|\geq |\out_H(S')|/8$, then, since all edges of $E_1(S)$ belong to $F^*$, $|\out_{F^*}(S)|\geq |\out_F(S)|/8$. We assume from now on that this is not the case, and so $|E_2(S)|\geq 7|\out_F(S)|/8$. Next, we construct a maximal set $E'\subseteq E_2(S)$ of edges, such that the edges in $E'$ do not share endpoints in graph $H$. This is done by a simple greedy algorithm: while $E_2(S)\neq \emptyset$, let $e\in E_2(S)$ be any edge. Add $e$ to $E'$, and delete from $E_2(S)$ edge $e$ and all edges sharing endpoints with $e$ in graph $H$. Since all edges in $E_2(S)$ are blue, and each vertex may be incident on at most two blue edges, for every edge added to $E'$, we delete at most three edges from $E_2(S)$. Therefore, eventually $|E'|\geq |E_2(S)|/3\geq 7|\out_H(S')|/24\geq |\out_H(S')|/4=|\out_F(S)|/4$ holds.

Each edge of $E'$ belongs to $\out_{F^*}(S)$ independently with probability at least $1/4$. The expected number of the edges of $E'$ that belong to $\out_{F^*}(S)$ is therefore at least $|E'|/4\geq |\out_F(S)|/16\geq N\cdot 2^{i-5}$. 

We use the following standard Chernoff bound: let $X_1,\ldots,X_n$ be independent random variables in $\set{0,1}$, and let $\mu=\expect{\sum_{i=1}^nX_i}$. Then $\prob{\sum_{i=1}^nX_i<\mu/2}\leq e^{-\mu/12}$. Therefore, the probability that $|\out_{F^*}(S)|<|\out_F(S)|/32$ is at most $e^{-N\cdot 2^{i-5}/12}$.
Overall, the probability that for some $S\in \cset_i$, $|\out_{F^*}(S)|<|\out_F(S)|/32$ is at most:

\[(n')^{2^{i+1}}\cdot e^{-2^{i-5}N/12}<1/(n')^2\]

since $N>1536 \log n'$. Using the union bound over all $1\leq i\leq \ceil{\log m}$, with probability at least $\frac{\ceil{\log m}}{(n')^2}$, for every set $S\subseteq V$, $|\out_{F^*}(S)|\geq |\out_F(S)|/32$. In particular, set $U$ is $\alphaWL/32$-well-linked in $F^*$ w.h.p. This concludes the proof of Claim~\ref{claim: well-linkedness in sampled graph}. As observed above, this implies that $A$ is $\alpha$-well-linked in graph $H^*$, thus completing the proof of Theorem~\ref{theorem: finding a minor in a super-cluster}. It now only remains to prove Claim~\ref{claim: large degree}.

\proofof{Claim~\ref{claim: large degree}}
We prove that the value of the minimum cut in $F$ is at least $N$.
Assume otherwise. Let $(X,Y)$ be a partition of $V(F)$, such that $X,Y\neq \emptyset$, and $|E_F(X,Y)|<N$. 
Let $X'\subseteq V(H)$ be obtained from $X$, by un-contracting all vertices $v_{\sigma}$: that is, $X'=\bigcup_{v_{\sigma}\in X}V(\sigma)$.
We construct $Y'$ from $Y$ similarly. Observe that $(X',Y')$ is a partition of $V(H)$, and $|E_H(X',Y')|<N$.


Assume first that there are two paths $P,P'\in \hset$, such that $P$ is contained in $H[X']$, and $P'$ is contained in $H[Y']$. We claim that $|E_H(X',Y')|\geq N$ in this case, leading to a contradiction. Indeed, recall that we have constructed $N$ expanders $X_1,\ldots,X_N$. For each $1\leq i\leq N$, expander $X_i$ contains some path $P_i$ connecting a pair $v,v'$ of vertices of $X_i$, where $v$ is embedded into a sub-path of $P$, and $v'$ is embedded into a sub-path of $P'$. Using the embedding of $X_i$ into $\tilde H_i$, path $P_i$ naturally defines a path $P'_i\subseteq \tilde H_i$, connecting a vertex of $P$ to a vertex of $P'$. It is immediate to see that paths $\set{P_i}_{i=1}^N$ are completely disjoint, as each such path is contained in a distinct graph $\tilde H_i$. Therefore, $H$ contains $N$ edge-disjoint paths connecting the vertices of $P$ to the vertices of $P'$. Each such path must contribute an edge to $E_H(X',Y')$, and so $|E_H(X',Y')|\geq N$, a contradiction.

Therefore, for one of the vertex sets $X',Y'$, no path $P\in \hset$ is 
contained in the sub-graph of $H$ induced by that set. 
We assume w.l.o.g. that this set is $X'$. 

Let $\rset$ be the set of paths, obtained from $\hset$, by deleting the edges of $E_H(X',Y')$ from them. Each path in $\rset$ is contained in either $H[X']$ or $H[Y']$, and we let $\rset'\subseteq \rset$ be the set of paths contained in $H[X']$. 
We claim that $|\rset'|<N$. Indeed, since no path of $\hset$ is contained in $H[X']$, every path $P\in \rset'$ contributes at least one edge to $E_H(X',Y')$. Consider now any path $P'\in \rset'$, and let $P\in \hset$ be the path such that $P'$ is a sub-path of $P$. Let $\sigma$ be any segment of $P'$, such that $v_{\sigma}\in X$. Since no path of $\hset$ is contained in $H[X']$, $\sigma$ must be a heavy segment of $P$. Therefore, there is some index $1\leq i\leq Nr^*$, such that the path $\sigma\cap H_i$ contains at least $200N^4$ vertices. We fix any such index $i^*$.

 We define a new set $\rset^*$ of paths as follows: for each path $P\in \rset'$, we add the path $P\cap H_{i^*}$ to $\rset^*$, if $P\cap H_{i^*}\neq \emptyset$. From the above discussion, $|\rset^*|< N$, and $|V(\rset^*)|\geq 200N^4$.

Recall that $\rset_{i^*}$ is the set of the red paths in $H_{i^*}$,
connecting $A_{i^*}$ to $B_{i^*}$, and $\bset_{i^*}$ is the set of
the blue paths in $H_{i^*}$, connecting $A_{i^*}'$ to $A_{i^*}''$
that we have constructed during the first step of the
algorithm. Clearly, each path in $\rset^*$ is a subpath of
a path in $\rset_{i^*}$.

Let $\bset$ be the set of paths, obtained from $\bset_{i^*}$, by deleting all edges of $E_H(X',Y')$ from them. 
Let $\bset^*\subseteq \bset$ be the set of paths contained in $H[X']$.
We claim that $|\bset^*|\leq 2N$. Indeed, recall that the paths of $\bset_{i^*}$ originate and terminate at the vertices of $A_{i^*}$. 
Since $|\rset'|< N$, at most $N$ such vertices $a\in A_{i^*}$ belong to $X'$. Therefore, at most $N$ paths of $\bset_{i^*}$ may be contained in $H[X']$. Every other path of $\bset^*$ must contribute one edge to $E_H(X',Y')$, and so in total $|\bset^*|\leq 2N$.

Observe that for every vertex $v\in V(\rset^*)$, if $v$ belongs to any blue path, then it belongs to a path in $\bset^*$. Similarly, for $v\in V(\bset^*)$, if $v$ belongs to any red path, then it belongs to a path in $\rset^*$. 

We will view the paths in $\rset_{i^*}$ as directed from $A_{i^*}$ to $B_{i^*}$, and we will view the paths in $\bset_{i^*}$ as directed from $A_{i^*}'$ to $A_{i^*}''$. Let $S_1$ be the set of all vertices $v$, such that $v$ is the first vertex on some path in $\rset^*$, and let $T_1$ be the set of all vertices $v$, such that $v$ is the last vertex on some path in $\rset^*$. We define the sets $S_2,T_2$ of vertices for $\bset^*$ similarly. Let $J=J(\rset^*\cup\bset^*)$. Then $|S_1|=|T_1|\leq N$, and $|S_2|=|T_2|\leq 2N$, while $|V(J)|\geq 200N^4$. Since every vertex in $V(H_{i^*})\setminus (A_{i^*}\cup B_{i^*})$ has degree at least $3$ in $H_{i^*}$, $\tau(J)\geq 200N^4-|S_1\cup S_2\cup T_1\cup T_2|\geq 200N^4-6N$.

From Theorem~\ref{thm: 2-flow-main}, there are sets $\rset'$ and $\bset'$ of paths, routing $(S_1,T_1)$ and $(S_2,T_2)$, respectively, in $J$, such that, if $J'=J(\rset'\cup \bset')$, then $\tau(J')<8(2N)^4+16N+1<200N^4-6N$. Since every vertex in $J\setminus (S_1\cup S_2\cup T_1\cup T_2)$ has degree more than $2$ in $J$, this means that there is some edge $e\in E(J)$, such that $(S_1,T_1)$ and $(S_2,T_2)$ are still routable in $J\setminus\set{e}$, via the sets $\rset'$ and $\bset'$ of paths.

We will now show that both $(A_{i^*},B_{i^*})$, and $(A_{i^*}',A_{i^*}'')$ remain routable in $H_{i^*}\setminus\set{e}$, contradicting the minimality of $H_{i^*}$. We show this for $(A_{i^*},B_{i^*})$; the proof for $(A_{i^*}',A_{i^*}'')$ is similar.

We start with a directed graph containing the original set $\rset_{i^*}$ of paths routing $(A_{i^*},B_{i^*})$, where the edges on these paths are oriented from $A_{i^*}$ towards $B_{i^*}$. We then delete from this graph all edges whose both endpoints are contained in $J$. 
Notice that the edge $e$ does not belong to the new graph. Also, for each vertex $v$:

\begin{itemize}

\item If $v\in A_{i^*}\setminus S_1$, then there is one edge leaving $v$ and no edges entering $v$;
\item If $v\in S_1\cap A_{i^*}$, then there is no edge entering or leaving $v$;
\item If $v\in S_1\setminus A_{i^*}$, then there is one edge entering $v$, and no edges leaving $v$;
\item if $v\in B_{i^*}\setminus T_1$, then there is one edge entering $v$ and no edges leaving $v$;
\item  If $v\in T_1\cap B_{i^*}$, then there is no edge entering or leaving $v$;
\item If $v\in T_1\setminus A_{i^*}$, then there is one edge leaving $v$, and no edges entering $v$;
\item for all other vertices $v$, either there is one edge entering and one edge leaving $v$, or there is no edge incident on $v$.
\end{itemize}

Finally, we add all edges lying on the paths in $\rset'$ to the
resulting graph. In this final graph, every vertex in $A_{i^*}$ has
one outgoing and no incoming edges, and every vertex in $B_{i^*}$ has
one incoming and no outgoing edges. Every other vertex either has
exactly one incoming and one outgoing edge, or it has no edges
incident on it. It is then easy to see that $(A_{i^*},B_{i^*})$ is
routable in this graph. Since this final graph is contained in
$H_{i^*}\setminus\set{e}$, this contradicts the minimality of
$H_{i^*}$.

\paragraph{Acknowledgement:} We thank Paul Seymour for posing the
question of the existence of degree-3 treewidth sparsifiers to us.


\bibliographystyle{alpha}
\bibliography{tw-sparsifier}

\newcommand{\etalchar}[1]{$^{#1}$}
\begin{thebibliography}{EGK{\etalchar{+}}10}

\bibitem[ACP87]{treewidth-np-hard}
S.~Arnborg, D.~Corneil, and A.~Proskurowski.
\newblock Complexity of finding embeddings in a k-tree.
\newblock {\em SIAM Journal on Algebraic Discrete Methods}, 8(2):277--284,
  1987.

\bibitem[And10]{Andrews}
Matthew Andrews.
\newblock Approximation algorithms for the edge-disjoint paths problem via
  {Raecke} decompositions.
\newblock In {\em Proceedings of IEEE FOCS}, pages 277--286, 2010.

\bibitem[BDFH09]{tw-no-poly-kernel}
Hans~L. Bodlaender, Rodney~G. Downey, Michael~R. Fellows, and Danny Hermelin.
\newblock On problems without polynomial kernels.
\newblock {\em Journal of Computer and System Sciences}, 75(8):423 -- 434,
  2009.

\bibitem[BK96]{BenczurK96}
Andr\'{a}s~A. Bencz\'{u}r and David~R. Karger.
\newblock Approximating s-t minimum cuts in $\tilde{O}(n^2)$ time.
\newblock In {\em Proceedings of the Twenty-eighth Annual ACM Symposium on
  Theory of Computing}, STOC '96, pages 47--55, New York, NY, USA, 1996. ACM.

\bibitem[Bod96]{Bodlaender-tw-fpt}
H.~Bodlaender.
\newblock A linear-time algorithm for finding tree-decompositions of small
  treewidth.
\newblock {\em SIAM Journal on Computing}, 25(6):1305--1317, 1996.

\bibitem[BSST13]{BatsonSST13}
Joshua~D. Batson, Daniel~A. Spielman, Nikhil Srivastava, and Shang-Hua Teng.
\newblock Spectral sparsification of graphs: theory and algorithms.
\newblock {\em Commun. ACM}, 56(8):87--94, 2013.

\bibitem[CC13]{CC13-grid}
Chandra Chekuri and Julia Chuzhoy.
\newblock Polynomial bounds for the grid-minor theorem.
\newblock {\em CoRR}, abs/1305.6577, 2013.
\newblock Extended abstract in {\em Proc.\ of ACM STOC}, 2014.

\bibitem[CE10]{mimicking3}
Erin~W. Chambers and David Eppstein.
\newblock Flows in one-crossing-minor-free graphs.
\newblock In Otfried Cheong, Kyung-Yong Chwa, and Kunsoo Park, editors, {\em
  ISAAC (1)}, volume 6506 of {\em Lecture Notes in Computer Science}, pages
  241--252. Springer, 2010.

\bibitem[CE13]{ChekuriE13}
Chandra Chekuri and Alina Ene.
\newblock Poly-logarithmic approximation for maximum node disjoint paths with
  constant congestion.
\newblock In {\em Proc.\ of ACM-SIAM SODA}, 2013.

\bibitem[Chu12a]{vsparsifiers}
Julia Chuzhoy.
\newblock On vertex sparsifiers with steiner nodes.
\newblock In {\em Proceedings of the 44th symposium on Theory of Computing},
  STOC '12, pages 673--688, New York, NY, USA, 2012. ACM.

\bibitem[Chu12b]{Chuzhoy11}
Julia Chuzhoy.
\newblock Routing in undirected graphs with constant congestion.
\newblock In {\em Proc.\ of ACM STOC}, pages 855--874, 2012.

\bibitem[CK09]{ChekuriK09}
Chandra Chekuri and Nitish Korula.
\newblock A graph reduction step preserving element-connectivity and
  applications.
\newblock In {\em Proc.\ of ICALP}, pages 254--265, 2009.

\bibitem[CL12]{ChuzhoyL12}
Julia Chuzhoy and Shi Li.
\newblock A polylogarithimic approximation algorithm for edge-disjoint paths
  with congestion 2.
\newblock In {\em Proc.\ of IEEE FOCS}, 2012.

\bibitem[CLLM10]{CLLM}
Moses Charikar, Tom Leighton, Shi Li, and Ankur Moitra.
\newblock Vertex sparsifiers and abstract rounding algorithms.
\newblock In {\em Proceedings of the 2010 IEEE 51st Annual Symposium on
  Foundations of Computer Science}, FOCS '10, pages 265--274, Washington, DC,
  USA, 2010. IEEE Computer Society.

\bibitem[Dru12]{Drucker12}
Andrew Drucker.
\newblock New limits to classical and quantum instance compression.
\newblock In {\em Foundations of Computer Science (FOCS), 2012 IEEE 53rd Annual
  Symposium on}, pages 609--618. IEEE, 2012.

\bibitem[EGK{\etalchar{+}}10]{EGK}
Matthias Englert, Anupam Gupta, Robert Krauthgamer, Harald R\"{a}cke, Inbal
  Talgam-Cohen, and Kunal Talwar.
\newblock Vertex sparsifiers: new results from old techniques.
\newblock In {\em Proceedings of the 13th international conference on
  Approximation, and 14 the International conference on Randomization, and
  combinatorial optimization: algorithms and techniques}, APPROX/RANDOM'10,
  pages 152--165, Berlin, Heidelberg, 2010. Springer-Verlag.

\bibitem[FHL08]{FeigeHL05}
U.~Feige, M.T. Hajiaghayi, and J.R. Lee.
\newblock Improved approximation algorithms for minimum weight vertex
  separators.
\newblock {\em SIAM Journal on Computing}, 38:629--657, 2008.

\bibitem[HNKR98]{mimicking0}
Torben Hagerup, Naomi Nishimura, Jyrki Katajainen, and Prabhakar Ragde.
\newblock Characterizing multiterminal flow networks and computing flows in
  networks of bounded treewidth.
\newblock {\em J. Comput. Syst. Sci.}, 57, 1998.

\bibitem[HO96]{element-connectivity}
H.~R. Hind and O.~Oellermann.
\newblock Menger-type results for three or more vertices.
\newblock {\em Congressus Numerantium}, 113:179--204, 1996.

\bibitem[Jan13]{Jansen13}
Bart M.~P. Jansen.
\newblock On sparsification for computing treewidth.
\newblock In {\em Proceedings of IPEC}, pages 216--229, 2013.

\bibitem[Kar99]{Karger-skeleton}
David~R. Karger.
\newblock Random sampling in cut, flow, and network design problems.
\newblock {\em Mathematics of Operations Research}, 24:383--413, 1999.

\bibitem[KR13]{mimicking1}
Robert Krauthgamer and Inbal Rika.
\newblock Mimicking networks and succinct representations of terminal cuts.
\newblock In {\em Proceedings of the Twenty-Fourth Annual ACM-SIAM Symposium on
  Discrete Algorithms}, pages 1789--1799. SIAM, 2013.

\bibitem[KRTV12]{mimicking2}
Arindam Khan, Prasad Raghavendra, Prasad Tetali, and L{\'a}szl{\'o}~A.
  V{\'e}gh.
\newblock On mimicking networks representing minimum terminal cuts.
\newblock {\em CoRR}, abs/1207.6371, 2012.

\bibitem[KRV09]{KRV}
Rohit Khandekar, Satish Rao, and Umesh Vazirani.
\newblock Graph partitioning using single commodity flows.
\newblock {\em J. ACM}, 56(4):19:1--19:15, July 2009.

\bibitem[KW12]{KratschW12}
Stefan Kratsch and Magnus Wahlstr{\"o}m.
\newblock Representative sets and irrelevant vertices: New tools for
  kernelization.
\newblock In {\em Proceedings of the 53rd Annual IEEE Symposium on Foundations
  of Computer Science}, FOCS '12, 2012.

\bibitem[LM10]{LM}
F.~Thomson Leighton and Ankur Moitra.
\newblock Extensions and limits to vertex sparsification.
\newblock In {\em Proceedings of the 42nd ACM symposium on Theory of
  computing}, STOC '10, pages 47--56, New York, NY, USA, 2010. ACM.

\bibitem[Lov76]{Lovasz-splitting-off}
L.~Lov\'asz.
\newblock On some connectivity properties of eulerian graphs.
\newblock {\em Acta Mathematica Academiae Scientiarum Hungarica},
  28(1-2):129--138, 1976.

\bibitem[LS12]{LeafS12}
Alexander Leaf and Paul Seymour.
\newblock Treewidth and planar minors.
\newblock Manuscript, available at
  https://web.math.princeton.edu/~pds/papers/treewidth/paper.pdf, 2012.

\bibitem[Mad78]{edge-connectivity}
W.~Mader.
\newblock A reduction method for edge connectivity in graphs.
\newblock {\em Ann. Discrete Math.}, 3:145--164, 1978.

\bibitem[MM10]{MM}
Konstantin Makarychev and Yury Makarychev.
\newblock Metric extension operators, vertex sparsifiers and lipschitz
  extendability.
\newblock In {\em FOCS}, pages 255--264. IEEE Computer Society, 2010.

\bibitem[Moi09]{Moitra}
Ankur Moitra.
\newblock Approximation algorithms for multicommodity-type problems with
  guarantees independent of the graph size.
\newblock In {\em FOCS}, pages 3--12. IEEE Computer Society, 2009.

\bibitem[OSVV08]{better-CMG}
Lorenzo Orecchia, Leonard~J. Schulman, Umesh~V. Vazirani, and Nisheeth~K.
  Vishnoi.
\newblock On partitioning graphs via single commodity flows.
\newblock In {\em Proceedings of the 40th annual ACM symposium on Theory of
  computing}, STOC '08, pages 461--470, New York, NY, USA, 2008. ACM.

\bibitem[Ree97]{Reed-chapter}
Bruce Reed.
\newblock {\em Surveys in Combinatorics}, chapter Treewidth and Tangles: A New
  Connectivity Measure and Some Applications.
\newblock London Mathematical Society Lecture Note Series. Cambridge University
  Press, 1997.

\bibitem[RS86]{RS-grid}
Neil Robertson and P~D Seymour.
\newblock {Graph minors. V. Excluding a planar graph}.
\newblock {\em Journal of Combinatorial Theory, Series B}, 41(1):92--114,
  August 1986.

\bibitem[RS10]{GraphMinors23}
Neil Robertson and Paul Seymour.
\newblock Graph minors {XXIII}: Nash-williams' immersion conjecture.
\newblock {\em Journal of Combinatorial Theory, Series B}, 100(2):181 -- 205,
  2010.

\bibitem[RST94]{RobertsonST94}
N~Robertson, P~Seymour, and R~Thomas.
\newblock {Quickly Excluding a Planar Graph}.
\newblock {\em Journal of Combinatorial Theory, Series B}, 62(2):323--348,
  November 1994.

\bibitem[RZ10]{RaoZhou}
Satish Rao and Shuheng Zhou.
\newblock Edge disjoint paths in moderately connected graphs.
\newblock {\em SIAM J. Comput.}, 39(5):1856--1887, 2010.

\bibitem[{Wol}13]{Wollan13}
P.~{Wollan}.
\newblock {The structure of graphs not admitting a fixed immersion}.
\newblock {\em ArXiv e-prints}, February 2013.

\end{thebibliography}

\end{document}